\documentclass[a4paper,12pt,smallextended]{svjour3} 

\usepackage{amsmath,graphicx,amssymb,fancyhdr,amsthm,verbatim}  
\newtheorem{thm}{Theorem}[section]  
  
\newtheorem{lem}[thm]{Lemma}

\theoremstyle{definition}  
  
\theoremstyle{remark}  
  
\usepackage{color}

\def\beq{\begin{eqnarray}}  
\def\eeq{\end{eqnarray}}  
\def\bsp{\begin{split}}  
\def\esp{\end{split}}

\def\d{\mathrm{d}}

\usepackage{hyperref}

\def\beq{\begin{eqnarray}}
\def\eeq{\end{eqnarray}}
\def\bem{\begin{matrix}}
\def\eem{\end{matrix}}

\def\bsp{\begin{split}}
\def\esp{\end{split}}

\def\d{\mathrm{d}}

\newcommand{\be}{\begin{equation}}
\newcommand{\ee}{\end{equation}}

\newcommand{\hnu}{\hat{\nu}}

\def \hbm #1 {\mbox{\boldmath{$\hat m^{(#1)}$}}}

\def \bm #1 {\mbox{\boldmath{$m^{(#1)}$}}}
\def \BDM {\begin{displaymath}}
\def \EDM {\end{displaymath}}

\usepackage[]{todonotes}
\usepackage{breqn}

\newcommand{\SigMinLam}{(\sigma-\lambda)}
\newcommand{\SigPluLam}{(\sigma+\lambda)}

\newcommand{\PsiTwo}{{\Psi}_2}

\newcommand{\deltaPsiTwo}{\delta \PsiTwo}
\newcommand{\DeltaPsiTwo}{\Delta \PsiTwo}

\newcommand{\deltaSigMinLam}{\delta \SigMinLam}

\newcommand{\deltadeltaPsiTwo}{\delta \delta \PsiTwo}
\newcommand{\DeltadeltaPsiTwo}{\Delta \delta \PsiTwo}

\newcommand{\DeltaDeltaPsiTwo}{\Delta \Delta \PsiTwo}

\begin{document}   
   
\title{\Large\textbf{On scalar curvature invariants in three dimensional spacetimes}   
\thanks{The final publication is available at Springer via \url{http://dx.doi.org/10.1007/s10714-016-2022-9}.}
}

\author{N. K. Musoke$^{1,2}$ \and D. D. McNutt$^{1}$ \and A. A. Coley$^{1}$ \and D. A. Brooks$^{1}$}
\date{Received: date / Accepted: date}
\institute{
	$^{1}$Department of Mathematics and Statistics,\\   
	Dalhousie University,    
	Halifax, Nova Scotia,\\    
	Canada B3H 3J5 \\
	Tel:  (902) 494 - 2572 \\ 
	Fax: (902) 494 - 5130 
	\vspace{0.1cm} \\   
	$^{2}$Perimeter Institute For Theoretical Physics,\\
	Waterloo, Ontario\\
	Canada N2L 2Y5
	\vspace{0.1cm}\\ 
	\texttt{msknathan@dal.ca}, \texttt{ddmcnutt@dal.ca}, \texttt{aac@mathstat.dal.ca}, \texttt{dario.a.brooks@dal.ca}
}
\titlerunning{Scalar Curvature Invariants}
\authorrunning{Musoke, McNutt, Coley and Brooks}
\maketitle   
   
\begin{abstract}   

We wish to construct  a minimal
set of algebraically independent scalar curvature invariants formed by the contraction of the Riemann (Ricci) tensor and its
covariant derivatives up to some order of differentiation in three dimensional (3D)
Lorentzian spacetimes. 
In order to do this we utilize the 
Cartan-Karlhede equivalence algorithm since, in general, all
Cartan invariants are related to scalar polynomial curvature invariants.
As an example we apply the algorithm to the class of 3D {{Szekeres}} cosmological
spacetimes with comoving
dust and cosmological constant $\Lambda$.
In this case,
we find that there are at most twelve
algebraically independent Cartan invariants, including $\Lambda$.
We present these Cartan invariants, and we relate them 
to twelve independent  scalar polynomial curvature invariants
(two, four and six, respectively, zeroth, first, and second order  scalar polynomial curvature invariants).


\end{abstract}   

\section{Introduction}

Scalar polynomial curvature invariants are scalars constructed from contractions of the Riemann tensor and its covariant
derivatives.  
Scalar polynomial curvature invariants have been utilised in the study of degenerate (Kundt) spacetimes, such as 
vanishing scalar invariant (VSI) and constant  scalar invariant (CSI) spacetimes
\cite{CGHP,CSI,applicat}.

In four dimensions (4D), scalar polynomial curvature invariants have been studied due to their potential use in general relativity (GR) \cite{exactSol,MacCallum:2015zaa}.  Scalar curvature invariants can be used to study the inequivalence of spacetimes  \cite{exactSol} and can be used to describe the physical properties of models in an invariant way \cite{AbdelqaderLake2015,PageShoom2015}. In addition they are utilized in higher order Lagrangians to generalize GR \cite{FKWC1992}. 

Any $\mathcal{I}$-non-degenerate spacetime
\cite{inv} 
can be completely characterized by scalar curvature invariants.  This leads to the
natural problem of attempting to find a basis for the scalar curvature invariants formed from the Riemann tensor
(up to some order of covariant differentiation).  We shall attempt to exploit the fact that the Cartan invariants give a complete local characterization of a manifold, and hence all scalar polynomial invariants depend on them. Therefore, in principle, we can utilize the Cartan invariants to construct a minimal set of algebraically independent scalar curvature invariants (including covariant derivatives of the curvature).

Primarily, the scalar curvature invariants formed from the Riemann tensor $R_{abcd}$ only (contractions involving products of
the undifferentiated Riemann tensor only, the so called \emph{algebraic invariants}) have been investigated in
4D Lorentzian spacetimes \cite{zak}.
Much less work has gone into studying curvature 
invariants formed using covariant derivatives of the Riemann tensor (differential invariants)
\cite{FKWC1992}.  We note that the independent Cartan invariants  at all orders of differentiation for spacetimes were specified in \cite{MacCallumAman1986}.

We are interested in the three-dimensional (3D) Lorentzian case \cite{Coley:2014goa}, 
in which the Riemann tensor reduces to the Ricci tensor (i.e., the Weyl tensor 
vanishes).
The main motivation for studying this problem in 3D
is mathematical, although there are some possible applications to
mathematical physics \cite{AlievNutku,Hall}. In particular, it may be applied to the initial value problem in GR; in this case a unique solution is given by a 3-metric and the 3D extrinsic curvature specified on a Cauchy hypersurface. Hence, an important step is to characterize the hypersurface using scalar curvature invariants. Related to this is the problem of matching across a 3D timelike hypersurface in 4D spacetime using invariants to obtain physically interesting exact solutions; for example, an attempt to match a material body (Vaidya metric) and an external vacuum (Kerr metric) across a 3D surface using both Cartan invariants and the Ricci scalar was made in \cite{Cox}. 

Our ultimate goal is to provide a basis of 
all such scalar curvature invariants in 3D  \cite{Coley:2014goa}. That is, we are interested in constructing  a minimal
set of algebraically independent scalar curvature invariants formed by 
the contraction of the Riemann (Ricci) tensor and its
covariant derivatives up to some maximum order of differentiation \footnote{Typically the maximum order of differentiation is written as $q=p+1$ where $p$ denotes the iteration of the CKA where the number of functionally independent invariants reaches a maximum and the dimension of the isotropy group of the $n-th$ covariant derivative of the curvature tensor for $n> p$  reaches its minimum.}  $q$ in 3D
Lorentzian spacetimes. 
There is a  general case (for which 
$q=1$) in which there are 18 algebraically independent invariants. In order to do this we exploit the 
Cartan-Karlhede  algorithm (CKA) for determining the equivalence
of manifolds since, in general, all
Cartan invariants are determined by scalar polynomial curvature invariants  \cite{inv}.

We apply the CKA to the class of 3D {{Szekeres}}
spacetimes which represent exact cosmological solutions with a comoving
dust and a cosmological constant $\Lambda$ \cite{BarrowEtAl}. In this particular example, 
$q=2$ and the invariant count, which records the number of functionally independent invariants at each iteration, is $(1,3,3)$. 
We find that there are at most twelve
algebraically independent Cartan invariants (including $\Lambda$) and present such a set. 
Then we relate these Cartan invariants to twelve independent polynomial scalar curvature invariants;
the two zeroth  order  polynomial scalar curvature invariants
$R, R_{ab}R^{ab}$, and
four and six, first and second order  scalar polynomial curvature invariants, respectively.

\subsection{Bounds on number of covariant derivatives in 3D}

We are interested in constructing  a minimal
set of algebraically independent scalar curvature invariants formed by 
the contraction of the Riemann (Ricci) tensor and its
covariant derivatives up to some maximum order of differentiation $q$ in 3D
Lorentzian spacetimes. 
The question of the bound (maximal value) for $q$ in 3D was addressed in \cite{Sousa}. 
In 3D spacetimes, the Weyl tensor vanishes and the canonical frame of the 
Karlhede algorithm \cite{karlhede} is aligned with principal directions of the Ricci tensor rather than the Weyl tensor. 
All spacetimes have $N_0\leq 1$ where $N_0$ is the dimension of the automorphism
group of the curvature tensor, hence we have $q \leq 5$.
{\footnote{The question of the maximal order of covariant derivative required for
the invariant classification of a $n$-dimensional pseudo-Riemannian manifold
$(\eta, M)$ is relevant in determining the worst case scenarios for
implementing the Cartan-Karlhede equivalence algorithm.  Cartan established the theoretical
upper-bound to be $q \leq n(n+1)/2$; that is, the dimension of $O(\eta,
M)$.  Karlhede improved the estimation of the upper bound to be $q \leq
n+N_0+1 < \dim O(\eta, M)$.  For the 3D Lorentzian
manifolds, this implies the upper-bound is at most five as $\dim O(\eta, M)=
\dim O(1,2)=6$.}}

It is common to discuss the Segre types of the traceless Ricci tensor \cite{exactSol}. The
case $\{(1,11)\}$ with full 3D Lorentz subgroup corresponding to spacetimes of constant curvature is of
Ricci type O.  There is a 2D isotropy group of the Lorentz group spanned by a boost and a null rotation.  The 2D
subgroups arise as special cases of Segre types $\{(1,1),1\}$ and $\{(2,1)\}$ 
with 1D isotropy group (with an additional isotropy).
However, these resulting spacetimes must be degenerate Kundt (i.e., not $\mathcal{I}$-non-degenerate), 
and are not considered further here. Suppose that the Karlhede algorithm takes five steps 
to complete (i.e., $q=4$). 
At most two of the parameters on the fibre bundle were fixed at the beginning, and so 
then the undetermined part of the Lorentz group has $N_0\geq 1$. We see that this 
argument does not improve the bound in 3D as it does in the 4D case \cite{Coley:2014goa}.

If the isotropy group is one dimensional, then $q$ can be 5. 
It was shown in \cite{MW2013} that
this bound is sharp (i.e., $q=5$ is attained; that is, a 3D Lorentzian manifold exists for
which the fifth covariant derivative of the Riemann tensor is required to
classify the spacetime invariantly).
However, in  the general cases that do not have a one dimensional
isotropy group (and are not degenerate Kundt), 
$q = p+1 \leq 4$. Thus, in these cases the
classification is simpler.

The cumulative number of algebraically independent scalars that can be constructed from the
Riemann tensor and its covariant derivatives up to order $m = n-2$ with $n \ge 2$ is given by the
Thomas formula $N(D,n)$ in $D$ dimensions \cite{exactSol,Thomas},
\begin{equation}
	N(D,n) 
	=
	\frac{D(D+1)(D+n)!}{2 D! n!}
	- \frac{(D+n+1)!}{(D-1)!(n+1)!}
	+ n
	.
	\label{}
\end{equation}
In 3D, we note that $N(3,0) = 3, N(3,1) = 18, N(3,2) = 45, N(3,3) = 87$, $N(3,4) = 147$, and $N(3,5) = 228$. 
The number of invariants gained for each successive differentiation is given by
\begin{equation}
	N(3,n)-N(3,n-1)
	=
	\frac{3}{2} (n+2)(n-1)
	,
	\label{}
\end{equation}
in analogy to the result $(n+1)(n+4)(n+5)$ known in 4D \cite{MacCallumAman1986,siklos1984relativistic}.
We can study 3D spacetimes of different algebraic types on a case by case basis.  It may be possible to lower the bounds on $q$ in some algebraic cases.
For example, there is a  general case in which $q=1$, in which case  $N(3,1) = 18$.

\subsection{The Cartan Karlhede equivalence algorithm in 3D}
 
For $\mathcal{I}$-non-degenerate spacetimes, all
Cartan invariants (obtained from the CKA) are determined by scalar polynomial curvature invariants  \cite{inv}.
Hence the equivalence method suggests a possible approach to determine a basis of 
all scalar curvature invariants.

The equivalence problem in 3D was first considered in \cite{Sousa}.
The maximal order of covariant derivative required for the 
invariant classification of a 3D Lorentzian pseudo-Riemannian manifold ($q \leq 5$ \cite{MW2013}) is relevant
in determining the worst case scenarios for implementing the equivalence algorithm.
It may be possible to provide a minimal basis on a case by case basis; some 
cases may have a smaller $q$ and will then have fewer algebraically independent scalar invariants.
The results have been applied to the equivalence of 3D
Godel-type spacetimes \cite{Sousa,McNutt}) and Segre type $\{(1,1),1\}$  spacetimes with no isotropy group \cite{Coley:2014goa}.

In the particular case of the (general) P-type I 3D spacetime with invariant
count (3,3), we have that $q=1$. At zeroth order there are 
3 independent scalar invariants (after using the Lorentz freedom). There are
eighteen scalar invariants involving the first covariant derivative, 
of which only 15 are algebraically
independent due to the Bianchi identities \cite{MW2013}. 
Therefore, in order to classify these spacetimes we only need at most (after choosing which three terms in the Bianchi identities are algebraically
dependent), 3+15 = 18 Cartan invariants .
The bases for scalar 
Riemann polynomials of order eight or less in the derivatives of the 
metric tensor and for tensorial Riemann polynomials of order six or 
less were displayed in~\cite{FKWC1992}.

\subsection{Newman-Penrose formalism in 3D}

In \cite{MW2013}, Milson and Wylleman (MW hereafter) introduced a 3D analogue to the Newman-Penrose
formalism in 4D by considering the 3D manifold
as geodesically embedded in a 4D Lorentzian manifold. It is possible to project the Ricci identities, Bianchi identities,
commutator relations and the effect of a change of basis on the spin
coefficients and curvature components to produce the 3D
equivalent, where instead of complex-valued spinors we work with
real-valued spinors.

As a simple byproduct of this formalism we may
refine the algebraic classification of the trace-free Ricci tensor and
relate the algebraic type to the zeroth order isotropy group of the
equivalence algorithm; i.e.,  the automorphism group of the Riemann tensor.
By doing this it was shown that the dimension of the zeroth order isotropy group $H_0$ is at most one, except in the case of the curvature Type {\bf O} spacetimes, which may be ignored
as these are locally homogeneous spacetimes.

Using this information MW identified those curvature types for which
the algorithm potentially requires five iterations by requiring that the
components of the traceless Ricci tensor $S_{ab}$ are constant and that the dimension of the isotropy
group is bounded by $0<\text{dim}H_0<3$.  This set consists of those
spacetimes of curvature type ${\bf D}, {\bf DZ}$ and ${\bf N}$ \footnote{These types correspond to the Segre type $\{(11),1\}, \{1(1,1)\}$ and $ \{(21)\}$ for the traceless Ricci tensors respectively.} with
constant coefficients for the components of the traceless Ricci tensor $S_{ab}$.  A spacetime for which the curvature
tensor and the $k$-th order covariant derivatives of the curvature tensor
have constant coefficients is \emph{curvature homogeneous of order $k$}
$(CH_k)$.  By examining the $CH_k$, $k>0$ conditions MW produce
conditions for a single family of spacetimes with the requirement that one
must compute the fifth covariant derivative of the curvature tensor in the
equivalence algorithm. In particular, they showed that \cite{MW2013}:
{\em{The order of a curvature-regular, 3D Lorentzian manifold is bounded by 
$q-1\leq 4$, that this bound is sharp, and every 4th order metric is locally isometric to
\begin{eqnarray} 
	2 (2Tx du + dw)^2 - 2 du(dx + adu); 
	a \equiv \frac{1-e^{4Tw}}{2T}+ (2T^2-C)(x-\delta_C)^2 + F(u) \nonumber 
\end{eqnarray}
\noindent where $x,u,w$ are local coordinates, $C,T$ are real constants such that 
$C+2T^2 \neq0$, $\delta_C = 1$ if $C=0$ and $\delta_C =0$ if $C =0$
and $F(u)$ is an arbitrary real function such that 
$(1+2TF(u))F'' \neq 3 T (F')^2,~ if C \neq 0; F' \neq 0,~if C =1$.
}}

\section{3D Szekeres spacetimes}

The 3D Szekeres spacetimes represent exact cosmological solutions with a comoving
dust and a cosmological constant $\Lambda$ \cite{BarrowEtAl}.
We shall apply the CKA to these {{Szekeres}} 3D
spacetimes, in which we may use the algebraic independence of the Cartan
invariants to determine the minimal number of algebraically
independent scalar curvature invariants. We identify the Cartan invariants that are
algebraically independent, and in the following section then relate them to  a set of algebraically
independent
polynomial scalar
curvature invariants.

In this example, the Segre Type of the traceless Ricci tensor is $\{ 1(1,1)\}$, 
or equivalently of `Petrov'-type DZ \cite{MW2013},
$q=2$ and the invariant count (which records the number of functionally independent invariants at each order) is 
(1,3,3) (where $\dim H_0=1$, $\dim H_1=0$, $\dim H_2=0$\footnote{Here we denote the isotropy group of the $n$-th covariant derivative of the curvature tensor as $H_n$.}).
We shall follow the notation of \cite{MW2013}.
We must first fix all frame freedom to ensure
that the coframe is invariant.

\subsection{The coframe and its dual frame; the connection coefficients}

We use the metric given by Barrow et al. \cite{BarrowEtAl} with the restrictions they specify for the 3D Szekeres case:
\begin{equation}
	\label{GenMetric}
	\d s^2
	=
	- \d t^2
	+ R^2(x,y,t) e^{\hnu(x,y)} \d x^2
	+ \frac{(R_{,y} + \hnu_{,y})^2}{S^2(y)F_{,y}^2	} \d y^2,
\end{equation}
where $R(t,y)$ and $\hat{\nu}(x,y)$ take the forms 
\begin{eqnarray} 
R(t,y) &=&  chL(t) + F(y)shL(t) = \cosh(\sqrt{\Lambda} t) + F(y) \frac{\sinh(\sqrt{\Lambda} t)}{\sqrt{\Lambda}}, \nonumber 
\\
\hat{\nu}(x,y) &=& - \ln(A(y) x^2 + 2B(y)x + C(y)), \nonumber 
\end{eqnarray}
and $\Lambda$ is a constant, $A(y), B(y) ,C(y), F(y)$ and $S(y)$ are arbitrary functions of one variable, 
up to the condition that $F_{,y} \neq 0$. Note that in order to avoid notational conflict with the spinor formalism of \cite{MW2013}, we use $\hnu$, $\hat{\rho}$ and $\hat{\kappa}$ where \cite{BarrowEtAl} uses $\nu$, $\rho$ and $\kappa$, respectively.

Using the formalism in \cite{MW2013} for the Szekeres metric 
\eqref{GenMetric} the coframe $\{\theta^a\}$ may be written in compact form:
\begin{eqnarray} -n_{~\mu} = & \theta^0_{~\mu} &= \left(dt - R e^{\hnu} dx\right),\nonumber \\
m_{~\mu}=& \theta^1_{~\mu}& =  \sqrt{2}\left(\frac{(R_{,y} + R \hnu_{,y})dy}{F_{,y}S}\right),  \nonumber \\
-\ell_{~\mu} = & \theta^2_{~\mu}& = \left(dt + R e^{\hnu} dx\right).  \nonumber \end{eqnarray}
The non-zero connection coefficients are 
{\footnote{We will choose to write the spin coefficients only 
in terms of $\gamma$, $\sigma-\lambda$, $\pi$, $\sigma+\lambda$.}}

\begin{eqnarray} 
& \kappa= \pi = -\tau = -\nu = \frac{F_{,y} S}{4R},~~\epsilon = -\gamma = \frac{R_{,t}}{4R}, & \nonumber \\
& \frac{\sigma-\lambda}{2} = \frac14 \frac{R_{,ty} + R_{,t} \hat{\nu}_{,y}}{R_{,y} + R \hat{\nu}_{,y}},~~ \frac{\sigma+ \lambda}{2} = - \frac14 \frac{e^{-\hat{\nu}} \hat{\nu}_{,xy}}{R_{,y} + R \hat{\nu}_{,y}}. & \nonumber 
\end{eqnarray}
Note that the frame derivatives $D = \ell^a \nabla_a$ and $\Delta = n^a \nabla_a$ are equal when applied to any scalar function not depending on $x$, such as $\gamma$. 

The commutator relations for the Szekeres spacetime are then 

\begin{eqnarray} & D \delta - \delta D = - \tau[ D -  \Delta ] + 2 \sigma \delta, & \label{Comm1} \\
& D \Delta - \Delta D = - 2 \gamma [ D -  \Delta ],  & \label{Comm2} \\
& \delta \Delta - \Delta \delta = - \tau[ D -  \Delta ] + 2 \lambda  \delta & \label{Comm3} \end{eqnarray}

\noindent The Segre Type of the traceless Ricci tensor is then $\{ 1(1,1)\}$, 
or equivalently of `Petrov'-type DZ \cite{MW2013}, since the non-zero components of the trace-free Ricci tensor are $\Psi_0 = \Psi_4 = 3\Psi_2$. These trace-free Ricci components are equivalent to the energy density $\hat{\kappa} \hat{\rho}$ of Barrow et al, whose form is given in the Szekeres case by their equations (87)-(88) with the restriction $F_{,x} = 0$:
\begin{eqnarray}
	S_{11} = \Psi_2 =    \frac{1}{12} \hat{\kappa} \hat{\rho} = \frac{ E(x,y; \Lambda)}{12R(R_{,y} + R \hat{\nu}_{,y})}, \nonumber 
\end{eqnarray}
where $\dot{R} = u^a R_{,a}$ for the given time-like vector $u^a$ and
\begin{align}
 E(x,y; \Lambda) &=& e^{-2\hnu} \left[ \hnu_{,xy} \hnu_{,x} - \hnu_{,xxy} \right] + \frac12 e^{-2\hnu}(e^{2\hnu}[\dot{R}^2 - \Lambda R^2 - 2 (S F_{,y})^2 ])_{,y} \nonumber 
 \\ 
&=& e^{-2\hnu} \left[ \hnu_{,xy} \hnu_{,x} - \hnu_{,xxy} \right] + \frac12 e^{-2\hnu}(e^{2\hnu}[- \Lambda + F^2 - 2 (S F_{,y})^2 ])_{,y}. 
\label{BarrowE}
\end{align}
The Ricci scalar is algebraically dependent on $\PsiTwo$  in the following manner: 
\begin{equation}
	R = 12 \Psi_2 + 3 \Lambda. \label{SzekereCurvCond}
\end{equation}

At zeroth order the frame is fixed up to null rotations and boosts; however it is invariant under rotations of the two spatial coordinates, given by equations (151)-(152) in \cite{MW2013}  with the corresponding effect on the curvature components and connection coefficients in (153)-(161).

\subsection{ The 0-th iteration of the  CKA: Segre types for the Szekeres Spacetimes}

\begin{lem}  \label{lem:ZerothOrder}
There are no 3D Szekeres spacetimes with $\Psi_2$ constant. 
\end{lem}

\begin{proof}

The Bianchi identities with $\Psi_2$ assumed to be constant, produce the following condition: 

\begin{equation}  3\Psi_2[2\gamma - \lambda + \sigma] = 0. \nonumber \end{equation}

\noindent From these equations it follows that $\sigma - \lambda = - 2\gamma = 2 \epsilon$. Equating $\sigma-\lambda$ with $2\epsilon$ we find that either $R_{,y} = 0$ or we find the following differential constraint: $C(y) R = R_{,y}$. However, for the particular form of $R$ this implies $C(y) = 0$ which produces the same result $R_{,y} = 0$ if and only if $F_{,y} = 0$ which cannot happen; thus $\Psi_2$ is non-constant. 
\end{proof}

The first five steps of the CKA for the $q$th iteration are then:
\begin{enumerate}
\item Calculate the set $I_q$, i.e. the derivatives of the curvature up to the $q$th order, i.e. the Ricci tensor. 
\item Fix the frame as much as possible by putting the elements of $I_q$ into canonical forms.
\item Find the frame freedom given by the isotropy group $H_q$ of transformations 
which leave invariant the normal form 
\footnote{What we are calling the normal forms of the Ricci tensor, Sousa et al. \cite{Sousa} 
 refer to as canonical forms.} of $I_q$.
\item Find the number $t_0$ of functionally independent functions of spacetime coordinates in the elements
of $I_q$, brought into the normal forms.
\item If the isotropy group $H_q$ is the same as 
$H_{q - 1}$ and the number of functionally independent functions $t_q$ is equal to $t_{q-1}$ , 
then let $q = p + 1$ and stop. [For $q=0$ there are no previous steps and we must always set $q=1$ and start 
on the 1st iteration of the CKA.]
\end{enumerate}

\noindent By studying the effect of the members of the Lorentz group
$SO(1,2)$ which do change the form of the metric, we may use these
frame transformations to fix the normal frame and determine all possible
Segre types.

\subsection{The 1st iteration of the CKA}

The next iteration of the CKA requires that we must compute the covariant derivative of the Ricci tensor 
in normal form and the scalar $\Lambda$. It is here where the connection coefficients (or spin coefficients 
if spinors are used) are introduced as potential invariants. This can be seen by using the formula 
for the covariant derivative:
\begin{eqnarray} R_{ab;c} = \nabla_c R_{ab}= R_{ab,c} - R_{db}\Gamma^d_{~ca} - R_{ad} \Gamma^d_{~cb},
\label{CovRic} \end{eqnarray}
\noindent introducing the following extended Cartan invariants at first order \footnote{We will call invariants constructed from combinations of Cartan invariants of different orders \emph{extended Cartan invariants}.}:

\begin{eqnarray} 3\Delta \Psi_2 + 12 \gamma \Psi_2,~\delta \Psi_2,~ 3 D \Psi_2 + 12 \gamma \Psi_2,~ (\sigma-\lambda) \Psi_2. \nonumber \end{eqnarray}

\noindent However, there is a problem: the connection coefficients used to compute the Ricci tensor prior to any frame transformations may be different from the coframe used to produce the normal form of the Ricci tensor.

We proceed to the next iteration of the algorithm. The metric coframe is not fully an invariant coframe, and so one must be careful with taking frame derivatives of invariants until an invariant coframe is discovered. To entirely fix the frame,  we consider the transformation rules for the spin-coefficients under an element of $SO(2)$:

\begin{eqnarray} (\ell+n)' = LHS,~~(\ell-n \pm 2im)' = e^{\mp it}(\ell-n \pm 2im) \nonumber\end{eqnarray}

\noindent where $LHS$ denotes the left-hand side of the equation without primes. This produces the following transformation rules for the spin coefficients and curvature components:

\begin{eqnarray} &(\gamma + \sigma - \epsilon - \lambda)' = LHS,& \nonumber \\ 
&(4\alpha+\kappa-\pi+\nu-\tau)' = LHS,& \nonumber \\
&(2(\gamma + \epsilon) \pm i(\kappa-\pi + \tau - \nu))' = e^{\pm i t} LHS, & \nonumber \\
&(4\alpha + \pi - \kappa + \tau - \nu + \pm 2o(\epsilon-\gamma+\sigma-\lambda))' = e^{\pm 2 it} LHS,& \nonumber \\
&(\lambda + \sigma - \gamma - \epsilon + \pm i(\pi - \tau))' = e^{\pm it}(LHS - \delta t \mp \frac{i}{2} (Dt - \Delta t)),& \nonumber \\
&(\kappa + \pi + \nu + \tau)' = LHS-(Dt + \Delta t),& \nonumber \\
&(\Psi_0 + 2\Psi_2 + \Psi_4)' = LHS, & \nonumber \\
&(\Psi_0 - \Psi_4 \pm 2i(\Psi_1 + \Psi_3))' = e^{\mp it} LHS, \nonumber \\
&(\Psi_0 - 6\Psi_2 + \Psi_4 \pm 4i (\Psi_1 - \Psi_3))' = e^{\mp 2 it} LHS,& \nonumber \end{eqnarray}

\noindent where $LHS$ denotes the left-hand side of the equation without primes. Substituting the non-zero spin-coefficients and Curvature components  \cite{MW2013}, we find the simpler transformation rules 

\begin{eqnarray} &(2\gamma + \sigma - \lambda)' = LHS,& \nonumber \\ 
&(4\alpha+\kappa-\pi+\nu-\tau)' = 0,& \nonumber \\
&(2(\gamma + \epsilon) \pm i(\kappa-\pi + \tau - \nu))' = 0, & \nonumber \\
&(4\alpha + \pi - \kappa + \tau - \nu + \pm 2i(\epsilon-\gamma+\sigma-\lambda))' = e^{\pm 2 it} (\sigma-\lambda),& \nonumber \\
&(\lambda + \sigma - \gamma - \epsilon + \pm i(\pi - \tau))' = e^{\pm it}(\lambda+\sigma \pm i (2\pi) - \delta t \mp \frac{i}{2} (Dt - \Delta t)),& \nonumber \\
&(\kappa + \pi + \nu + \tau)' = -(Dt + \Delta t),& \nonumber \\
&(\Psi_0 + 2\Psi_2 + \Psi_4)' = 0, & \nonumber \\
&(\Psi_0 - \Psi_4 \pm 2i(\Psi_1 + \Psi_3))' = 0, \nonumber \\
&(\Psi_0 - 6\Psi_2 + \Psi_4 \pm 4i (\Psi_1 - \Psi_3))' = 0.& \nonumber \end{eqnarray}

\noindent From these equations it is clear that $t=0$ produces the smallest set of first order  invariants. 

Thus at first order the extended Cartan invariants are $D \Psi_2, \delta \Psi_2, \Delta \Psi_2, \gamma$ and $\sigma- \lambda$.  Due to the simplifications for the spin-coefficients there is only one independent Bianchi identity: 

\begin{eqnarray}
	D\Psi_2 + \Delta \Psi_2 = 2 [2\gamma + (\sigma-\lambda)] \Psi_2 \label{Bianchi} 
\end{eqnarray} 

\noindent With this identity, we may write one of the first order invariants as an algebraic combination of the remaining first order invariants and the zeroth order invariants. While this choice is arbitrary, we have chosen to eliminate $D \Psi_2$ as an algebraically independent invariant.
 
At first order there are only five (apart from $\Lambda$) algebraically independent extended Cartan invariants ($\Psi_2$ at zeroth order, plus $\gamma$, $\sigma-\lambda$, $\delta\Psi_2$ and $\Delta\Psi_2$ at first order).  
The number of functionally independent invariants is at least two and at most three.  
We will assume that we are in the general case with invariant count (1,3,3)  to provide a complete example of the equivalence algorithm in practice.

\subsection{The 2nd iteration of the CKA}

At second order the remaining non-zero invariants in terms of spin-coefficients appear, i.e. $\sigma + \lambda$ and $\pi$,  along with the twelve frame derivatives of the first order invariants.
Six of these are the frame derivatives of the spin-coefficients $\gamma$ and $\sigma-\lambda$.
The non-trivial Ricci identities (equations (107)-(115) of \cite{MW2013}) allow for five of these to be written as algebraic expressions of the other invariants:
\begin{eqnarray} 
D \sigma + \delta \tau &=& 2\tau^2 + 2(-\gamma + \sigma)\sigma + \frac32 \Psi_2, \label{RicciId0}
\\
D \tau + \Delta \tau &=& 4 \gamma \tau, \label{RicciId1}
\\
\delta \gamma &=& -2\gamma \pi + (\lambda-\sigma)\tau, \label{RicciId2}
\\
\delta \tau - \Delta \sigma &=& 2(\lambda - \gamma )\sigma + 2\tau^2 + \frac32 \Psi_2 + \frac{\Lambda}{4}, \label{RicciId3}
\\
D\gamma + \Delta \gamma &=& 4 \gamma^2 - \frac{\Lambda}{4}, \label{RicciId4}
\\
D \lambda + \delta \tau &=& 2(\sigma + \gamma )\lambda + 2 \tau^2 + \frac32 \Psi_2 +\frac{\Lambda}{4}, \label{RicciId5}
\\
\delta \tau - \Delta \lambda &=& 2\tau^2 + 2 (\gamma + \lambda)\lambda + \frac32 \Psi_2. \label{RicciId8}
\end{eqnarray}
Combining the pairs of equations [\eqref{RicciId0}, \eqref{RicciId5}] and [\eqref{RicciId3}, \eqref{RicciId8}] we produce the equations for $D(\sigma-\lambda)$ and $\Delta(\sigma-\lambda)$. 
Equation \eqref{RicciId2} gives us $\delta\gamma$ and equation \eqref{RicciId4} gives us $D\gamma$ in terms of $\Delta \gamma$.
However, as $\gamma$ lacks any $x$-dependence, $\Delta \gamma = D \gamma$ and so these may be eliminated 
as well. 
Thus, of the six frame derivatives of the first order invariants $\gamma$ and $\sigma-\lambda$, only $\delta(\sigma-\lambda)$ is independent.

Taking the frame derivatives of the non-zero Bianchi identity, \eqref{Bianchi}, and applying the commutator relations, \eqref{Comm1}-\eqref{Comm3},  we may reduce the number of algebraically independent invariants at second order to six.
To see this, if we continue with the choice of $D\Psi_2$ being the algebraically dependent invariant at first order, we may eliminate $D^2 \Psi_2, D\delta \Psi_2$, $D\Delta \Psi_2$ and $\delta\Delta\Psi_2$
Thus at second order, the six remaining algebraically independent invariants are: $\pi, \sigma+\lambda$, $\delta(\sigma-\lambda), \delta^2 \Psi_2, \Delta^2 \Psi_2$ and $\Delta \delta \Psi_2$. 

The algorithm ends at second order as the number of functionally
independent invariants and the dimension of the isotropy group is unchanged
from the previous iteration.  This implies that there are at most twelve
algebraically independent extended Cartan invariants including $\Lambda$. Of course the full
classification also includes all of the discrete information (e.g, the invariant count),
and the relationships between the Cartan invariants and all of the vanishing invariants.

\section{Scalar polynomial invariants}
\label{sec:PolynomialInvariants}

Next we relate the above Cartan invariants to some of the polynomial scalar curvature invariants.
At zeroth and first order, we use the overdetermined basis of scalar curvature invariants from section 3 of \cite{Coley:2014goa}.
At second order, we look at invariants from~\cite{FKWC1992}.

In the previous sections, the Cartan invariants were calculated.
Some of these are algebraically dependent. 
The details of the algebraic dependences of the Cartan invariants are described more fully in the following subsections as we eliminate the independent invariants from the Ricci tensor and its covariant derivatives.
For clarity, our choice of algebraically independent Cartan invariants is summarized here. 
\\
0th order: $\Psi_2$ ($\Lambda$)
\\
1st order: $\gamma$, $\sigma-\lambda$, $\delta\Psi_2$, $\Delta\Psi_2$
\\
2nd order: $\pi$, $\sigma+\lambda$, $\delta(\sigma-\lambda)$, $\delta\delta\Psi_2$, $\Delta\Delta\Psi_2$, $\Delta\delta\Psi_2$

\subsection{0th order}
\label{sec:ZerothOrderPolynomialInvariants}

The Ricci Tensor is given as a function of the 0th order Cartan invariant $\Psi_2$ and the constant $\Lambda$ above:
\begin{align}
	R_{ij} =&\; S_{ij} + \frac{R}{3} g_{ij}
	\\
	=&\; 3 \Psi_2 \; \ell_i \ell_j 
	- \left( \frac{\Psi_2}{3} + 2 \Lambda \right) \; \ell_{(i} n_{j)}
	+ \left( \frac{5 \Psi_2}{3} + \Lambda \right) \; m_{i} m_{j}
	+ 3 \Psi_2 \; m_i m_j. 
\end{align}
where $S_{ij}$ is the tracefree Ricci tensor.
From here we may calculate some zeroth order polynomial scalar invariants. 
For example,
\begin{align}
	R = R^a_{~a}
	&= 
	12\,{ \PsiTwo}+3\,\Lambda
	\\
	R^{ab} R_{ab} &=
	3 \left( 24 \Psi_2^2 + 8 \Psi_2 \Lambda + \Lambda^2 \right).
	\label{eqn:ZerothOrderPSCIs}
\end{align}
Note that these two zeroth order invariants are polynomials in the zeroth order Cartan invariants.
This suggests an equivalence between these two polynomial scalar curvature invariants and the 0th order Cartan invariants $\Psi_2$ and $\Lambda$.
Indeed,
\begin{align}
	{\Psi_2}^2 &= \frac{3 R^{ab} R_{ab} - R^2}{72},
	\\
	\Lambda &= \frac{R}{3} - 4 \Psi_2.
\end{align}

\subsection{1st order}

The covariant derivative of the Ricci tensor can be calculated. 
Performing this derivative introduces both spin-coefficients (connection coefficients) and frame derivatives of the components at zeroth order (i.e. frame derivatives of $\Psi_2$).

As mentioned above, the Bianchi identity (11) (derived from $R^\mu_{\eta;\mu} = \frac{1}{2} R_{;\eta}$ or equations (116)-(118) of \cite{MW2013}) gives frame derivatives of $\Psi_2$ as combinations of the zeroth and first order invariants.
In particular,
\begin{align}
	D \Psi_2 &= 2 \Psi_2 (\sigma-\lambda) + 4 \gamma \Psi_2 - \Delta \Psi_2,
\end{align}
can be used to eliminate $D \PsiTwo$ from the components $R_{ij;k}$ of the first covariant derivative of the Ricci tensor. Thus its components can be expressed as polynomials in the algebraically independent 0th and 1st order Cartan invariants.
\begin{align}
	R_{ij;k} 
	&=
	\big( 24 \gamma \Psi_2 + 6 \Psi_2 \SigMinLam - 3 \DeltaPsiTwo \big) 
	\; \ell_{(i} \ell_{j)} \ell_k
	+
	\big( 3 \deltaPsiTwo \big) 
	\; \ell_{(i} \ell_{j)} m_k
	\\&\quad\nonumber
	+
	\big( -12 \gamma \Psi_2 + 3 \DeltaPsiTwo \big) 
	\; \ell_{(i} \ell_{j)} n_k
	+
	\big( 6 \Psi_2 \SigMinLam \big) 
	\ell_{(i} m_{j)} m_k 
	\\&\quad\nonumber
	+
	\big( -12 \Psi_2 \SigMinLam - 24 \gamma \Psi_2 + 6 \DeltaPsiTwo \big) 
	\; \ell_{(i} n_{j)} \ell_k 
	+
	\big( -6 \deltaPsiTwo \big) 
	\; \ell_{(i} n_{j)} m_k 
	\\&\quad\nonumber
	+
	\big( -6 \DeltaPsiTwo \big) 
	\; \ell_{(i} n_{j)} n_k 
	+
	\big( 12 \gamma \Psi_2 + 6 \Psi_2 \SigMinLam - 3 \DeltaPsiTwo \big) 
	\; m_{(i} m_{j)} \ell_k
	\\&\quad\nonumber
	+
	\big( 3 \deltaPsiTwo \big)
	\; m_{(i} m_{j)} m_k
	+
	\big( 3 \DeltaPsiTwo \big)
	\; m_{(i} m_{j)} n_k
	\\&\quad\nonumber
	+
	\big( 6 \Psi_2 \SigMinLam \big) 
	\; m_{(i} m_{j)} \ell_k  
	+ 
	\big( 6 \Psi_2 \SigMinLam - 3 \DeltaPsiTwo \big)
	\; n_{(i} n_{j)} \ell_k
	\\&\quad\nonumber
	+ 
	\big( 3 \deltaPsiTwo \big)
	\; n_{(i} n_{j)} m_k
	+ 
	\big( 12 \gamma \Psi_2 + 3 \DeltaPsiTwo \big)
	\; n_{(i} n_{j)} m_k
\end{align}

This allows us to calculate some 1st order polynomial scalar invariants which include the extended Cartan invariants appearing at 1st order.
These invariants are the 4-7th and 9th from the list in \cite{Coley:2014goa}.
\begin{dmath}
	R^{;a} R_{;a} =
-1152\,{ \DeltaPsiTwo}\,{ \PsiTwo}
\,{ \gamma}-576\,{ \DeltaPsiTwo}\,
{ \PsiTwo}\,{ \SigMinLam}+288\,{{
 \DeltaPsiTwo}}^{2}+288\,{{ 
\deltaPsiTwo}}^{2}
\end{dmath}
\begin{dmath}
	R^{bc;d} R_{bc;d} =
-576\,{{ \PsiTwo}}^{2}{{ \gamma}}^
{2}-144\,{{ \PsiTwo}}^{2}{{ 
\SigMinLam}}^{2}-576\,{ \DeltaPsiTwo}
\,{ \PsiTwo}\,{ \gamma}-288\,{ 
\DeltaPsiTwo}\,{ \PsiTwo}\,{ 
\SigMinLam}+144\,{{ \DeltaPsiTwo}}^{2
}+144\,{{ \deltaPsiTwo}}^{2}
\end{dmath}
\begin{dmath}
	R^{bc;d} R_{bd;c} =
-576\,{{ \PsiTwo}}^{2}{{ \gamma}}^
{2}-288\,{ \gamma}\,{{ \PsiTwo}}^{
2}{ \SigMinLam}-144\,{{ \PsiTwo}}^
{2}{{ \SigMinLam}}^{2}-288\,{ 
\DeltaPsiTwo}\,{ \PsiTwo}\,{ 
\gamma}-144\,{ \DeltaPsiTwo}\,{ 
\PsiTwo}\,{ \SigMinLam}+72\,{{ 
\DeltaPsiTwo}}^{2}+72\,{{ 
\deltaPsiTwo}}^{2}
\end{dmath}
\begin{dmath}
	R^{;c} R_c^{~e;f} R_{;e} R_{;f} =
-1327104\,{{ \PsiTwo}}^{4}{{ 
\gamma}}^{4}-1990656\,{{ \PsiTwo}}^{4
}{ \SigMinLam}\,{{ \gamma}}^{3}-
995328\,{{ \PsiTwo}}^{4}{{ 
\SigMinLam}}^{2}{{ \gamma}}^{2}-
165888\,{{ \PsiTwo}}^{4}{{ 
\SigMinLam}}^{3}{ \gamma}+663552\,{
 \DeltaPsiTwo}\,{{ \PsiTwo}}^{3}{{
 \gamma}}^{3}+331776\,{ 
\DeltaPsiTwo}\,{{ \PsiTwo}}^{3}{ 
\SigMinLam}\,{{ \gamma}}^{2}-165888\,
{ \DeltaPsiTwo}\,{{ \PsiTwo}}^{3}{
{ \SigMinLam}}^{2}{ \gamma}-82944
\,{ \DeltaPsiTwo}\,{{ \PsiTwo}}^{3
}{{ \SigMinLam}}^{3}+497664\,{{ 
\DeltaPsiTwo}}^{2}{{ \PsiTwo}}^{2}{{
 \gamma}}^{2}+165888\,{{ \PsiTwo}}
^{2}{{ \gamma}}^{2}{{ \deltaPsiTwo
}}^{2}+663552\,{{ \DeltaPsiTwo}}^{2}{{ 
\PsiTwo}}^{2}{ \SigMinLam}\,{ 
\gamma}+207360\,{{ \DeltaPsiTwo}}^{2}
{{ \PsiTwo}}^{2}{{ \SigMinLam}}^{2
}-41472\,{{ \PsiTwo}}^{2}{{ 
\SigMinLam}}^{2}{{ \deltaPsiTwo}}^{2}
-331776\,{{ \DeltaPsiTwo}}^{3}{ 
\PsiTwo}\,{ \gamma}-331776\,{ 
\DeltaPsiTwo}\,{ \PsiTwo}\,{ 
\gamma}\,{{ \deltaPsiTwo}}^{2}-165888
\,{{ \DeltaPsiTwo}}^{3}{ \PsiTwo}
\,{ \SigMinLam}-165888\,{ 
\DeltaPsiTwo}\,{ \PsiTwo}\,{ 
\SigMinLam}\,{{ \deltaPsiTwo}}^{2}+
41472\,{{ \DeltaPsiTwo}}^{4}+82944\,{{ 
\DeltaPsiTwo}}^{2}{{ \deltaPsiTwo}}^{
2}+41472\,{{ \deltaPsiTwo}}^{4}
\end{dmath}
\begin{dmath}
	R^{;c} R_c^{~e;f} R_{e~;f}^{~h} R_{;h} =
-663552\,{{ \PsiTwo}}^{4}{{ \gamma
}}^{4}-995328\,{{ \PsiTwo}}^{4}{ 
\SigMinLam}\,{{ \gamma}}^{3}-414720\,
{{ \PsiTwo}}^{4}{{ \SigMinLam}}^{2
}{{ \gamma}}^{2}+20736\,{{ \PsiTwo
}}^{4}{{ \SigMinLam}}^{4}+663552\,{ 
\DeltaPsiTwo}\,{{ \PsiTwo}}^{3}{{ 
\gamma}}^{3}+331776\,{ \DeltaPsiTwo}
\,{{ \PsiTwo}}^{3}{ \SigMinLam}\,{
{ \gamma}}^{2}-82944\,{ 
\DeltaPsiTwo}\,{{ \PsiTwo}}^{3}{{ 
\SigMinLam}}^{2}{ \gamma}-41472\,{
 \DeltaPsiTwo}\,{{ \PsiTwo}}^{3}{{
 \SigMinLam}}^{3}+165888\,{{ 
\DeltaPsiTwo}}^{2}{{ \PsiTwo}}^{2}{{
 \gamma}}^{2}+82944\,{{ \PsiTwo}}^
{2}{{ \gamma}}^{2}{{ \deltaPsiTwo}
}^{2}+331776\,{{ \DeltaPsiTwo}}^{2}{{ 
\PsiTwo}}^{2}{ \SigMinLam}\,{ 
\gamma}+103680\,{{ \DeltaPsiTwo}}^{2}
{{ \PsiTwo}}^{2}{{ \SigMinLam}}^{2
}-41472\,{{ \PsiTwo}}^{2}{{ 
\SigMinLam}}^{2}{{ \deltaPsiTwo}}^{2}
-165888\,{{ \DeltaPsiTwo}}^{3}{ 
\PsiTwo}\,{ \gamma}-165888\,{ 
\DeltaPsiTwo}\,{ \PsiTwo}\,{ 
\gamma}\,{{ \deltaPsiTwo}}^{2}-82944
\,{{ \DeltaPsiTwo}}^{3}{ \PsiTwo}
\,{ \SigMinLam}-82944\,{ 
\DeltaPsiTwo}\,{ \PsiTwo}\,{ 
\SigMinLam}\,{{ \deltaPsiTwo}}^{2}+
20736\,{{ \DeltaPsiTwo}}^{4}+41472\,{{ 
\DeltaPsiTwo}}^{2}{{ \deltaPsiTwo}}^{
2}+20736\,{{ \deltaPsiTwo}}^{4}
\end{dmath}
These invariants are polynomials in the six 0th and 1st order extended  Cartan invariants $\Psi_2$, $\Lambda$, $\sigma-\lambda$, $\gamma$, $\Delta\Psi_2$, and $\delta\Psi_2$.
The 0th order invariants $\Psi_2$ and $\Lambda$ have already been expressed in terms of the scalar curvature invariants, so of these six, four have not previously been expressed in such a manner.
However, we do find that taking the algebraic combinations 
\begin{dmath}
	I_{R1} {:=} {\frac {R^{bc;d} R_{bc;d}}{144}} - {\frac {R^{;a} R_{;a}}{288}} 
	=
-4\,{{ \PsiTwo}}^{2}{{ \gamma}}^{2
}-{{ \PsiTwo}}^{2}{{ \SigMinLam}}^
{2}
\end{dmath}
\begin{dmath}
	I_{R2} {:=} 
	{\frac {R^{bc;d} R_{bd;c}}{72}} + {\frac {R^{;a} R_{;a}}{288}} - {\frac {R^{bc;d} R_{bc;d}}{72}}
	=
-4\,{ \gamma}\,{{ \PsiTwo}}^{2}{
 \SigMinLam}
\end{dmath}
we can write manageable expressions relating $\gamma$ and $\sigma-\lambda$ to scalar invariants.
Then the two combinations
\begin{dmath}
	\frac{R^{;a} R_{;a}}{288} =
-4\,{ \DeltaPsiTwo}\,{ \PsiTwo}\,{
 \gamma}-2\,{ \DeltaPsiTwo}\,{ 
\PsiTwo}\,{ \SigMinLam}+{{ 
\DeltaPsiTwo}}^{2}+{{ \deltaPsiTwo}}^
{2}
	\label{}
\end{dmath}
and
\begin{dmath}
	\frac{1}{41472}
	\bigl(
		R^{;c} R_c^{~e;f} R_{;e} R_{;f} 
		-2 
		R^{;c} R_c^{~e;f} R_{e~;f}^{~h} R_{;h} 
	\bigr)
	+
	\frac{I_{R2}^2}{4}
	=
-4\,{{ \PsiTwo}}^{4}{{ \SigMinLam}
}^{3}{ \gamma}-{{ \PsiTwo}}^{4}{{
 \SigMinLam}}^{4}-16\,{ 
\DeltaPsiTwo}\,{{ \PsiTwo}}^{3}{{ 
\gamma}}^{3}-8\,{ \DeltaPsiTwo}\,{{
 \PsiTwo}}^{3}{ \SigMinLam}\,{{
 \gamma}}^{2}+4\,{{ \DeltaPsiTwo}}
^{2}{{ \PsiTwo}}^{2}{{ \gamma}}^{2
}+{{ \PsiTwo}}^{2}{{ \SigMinLam}}^
{2}{{ \deltaPsiTwo}}^{2}
\end{dmath}
are two algebraically independent expressions containing $\delta\Psi_2$ and $\Delta\Psi_2$.

This gives us a correspondence between the 0th and 1st order extended Cartan invariants and scalar curvature invariants selected from \cite{Coley:2014goa}.

\subsection{2nd order}

We would hope to be able to do something similar with the 2nd order invariants.
If possible, the result would be a set of PSCIs sufficient for classification of the space.

As noted in section 2.5 above, the Ricci identities can be used to write five of the six frame derivatives of the spin coefficients appearing at 1st order as algebraic combinations of the invariants $\Lambda$, $\Psi_2$, $\gamma$, $\sigma-\lambda$, $\pi$, $\sigma+\lambda$, and $\delta(\sigma-\lambda)$.
Then the explicit relations described after \eqref{RicciId8} are
\begin{align}
	D(\sigma-\lambda) 
	&=
	2(-\gamma+\sigma)\sigma-2(\sigma+\gamma)\lambda-\frac{\Lambda}{4},
	\\
	\Delta(\sigma-\lambda) 
	&=
	-2(\lambda-\gamma)\sigma+2(\gamma+\lambda)\lambda-\frac{\Lambda}{4},
	\\
	\delta\gamma 
	&=
	(-2\gamma+\sigma-\lambda)\pi,
	\\
	D\gamma = \Delta\gamma = \frac{D\gamma+\Delta\gamma}{2}
	&=
	2 \gamma^2 - \frac{\Lambda}{8}.
\end{align}
We must also make use of frame derivatives of the Bianchi identity \eqref{} and the commutation relationships \eqref{Comm1}-\eqref{Comm3} to eliminate $D^2 \Psi_2$, $D \delta \Psi_2$, $D \Delta \Psi_2$, and $\delta \Delta \Psi_2$.
This allows the second covariant derivative of the Ricci tensor to be written as a function of the independent invariants $\Lambda$, $\Psi_2$, $\gamma$, $\sigma-\lambda$, $\delta\Psi_2$, $\Delta\Psi_2$, $\pi$, $\sigma+\lambda$, and $\delta(\sigma-\lambda)$, $\delta^2 \Psi_2$, $\Delta^2 \Psi_2$ and $\Delta\delta \Psi_2$ only.
Each one of the expected invariants does in fact appear in the explicit expression.


We now desire some set of polynomial scalar curvature invariants which is
equivalent to the second order extended Cartan invariants.  In order to determine
these, we consult \cite{FKWC1992} for lists of (generally) independent
invariants constructed from second derivatives of the Riemann and Ricci
tensors.  In the notation of \cite{FKWC1992}, ${\cal R}^r_{s,q}$ denotes
the space of Riemann polynomials of rank $r$ (number of free indices),
order $s$ (number of differentiations of the metric tensor) and degree $q$
(number of factors $\nabla^p R_{\dots}^{\dots}$).  The notation ${\cal
R}^r_{\lbrace{\lambda_1 \dots \rbrace}}$ denotes the space of Riemann
polynomials of rank $r$ spanned by contractions of products of the form
$\nabla^{\lambda_1}R_{\dots}^{\dots}$.  

As we are interested in scalars, we consider the sets with rank $r=0$, contained in \cite{FKWC1992} Appendix~B.
After some consideration, of the sets presented there, only $\mathcal{R}^0_{4,1}$, $\mathcal{R}^0_{\{2,0\}}$, $\mathcal{R}^0_{\{2,2\}}$ and $\mathcal{R}^0_{\{2,0,0\}}$ are useful when comparing with second order extended Cartan invariants.
Of these, we may also eliminate those with minimal dimension greater than 3 and those which are not constructed from 2nd derivatives of the Riemann tensor.
There is the possibility that invariants higher order in the metric are needed, but they are not included in FKWC.
For convenience, we label the FKWC invariants by $I_{s,\alpha,i}$, where $s$ is the order, $\alpha=a,b,c,\dots$ enumerates the set within that order, and $i=1,2,3\dots$ enumerates the invariant within that set.

Beginning with $\mathcal{R}^0_{4,1}$, we have 
\begin{dmath}
	I_{4,a,1}
	=
	\Box R
	=
144\,{ \PsiTwo}\,{ \gamma}\,{ 
\SigMinLam}-144\,{ \PsiTwo}\,{ 
\gamma}\,{ \SigPluLam}-24\,{ 
\PsiTwo}\,{{ \SigMinLam}}^{2}+24\,{
 \PsiTwo}\,{ \SigMinLam}\,{ 
\SigPluLam}+24\,\Lambda\,{ \PsiTwo}-
144\,{ \gamma}\,{ \DeltaPsiTwo}+48
\,{ \pi}\,{ \deltaPsiTwo}-48\,{
 \SigMinLam}\,{ \DeltaPsiTwo}+24\,
{ \DeltaPsiTwo}\,{ \SigPluLam}+24
\,{ \DeltaDeltaPsiTwo}+24\,{ 
\deltadeltaPsiTwo}
	\label{}
\end{dmath}
which contains four of the second order extended Cartan invariants: $\pi$, $\sigma+\lambda$, $\Delta\Delta\Psi_2$, and $\delta\delta\Psi_2$.

Moving on to $\mathcal{R}^0_{\{2,0\}}$, we find that the invariants are longer expressions, although only in the same variables as $I_{a,4,1}$.
For example, the first is
\begin{dmath}
	I_{6,b,2}
	=
	R^{;pq} R_{pq}
	=
144\,\Lambda\,{ \PsiTwo}\,{ \gamma
}\,{ \SigMinLam}-144\,\Lambda\,{ 
\SigPluLam}\,{ \PsiTwo}\,{ 
\gamma}-24\,\Lambda\,{ \PsiTwo}\,{{
 \SigMinLam}}^{2}+24\,\Lambda\,{ 
\SigPluLam}\,{ \PsiTwo}\,{ 
\SigMinLam}+1152\,{{ \PsiTwo}}^{2}{{
 \gamma}}^{2}+1440\,{{ \PsiTwo}}^{
2}{ \gamma}\,{ \SigMinLam}-864\,{
 \SigPluLam}\,{{ \PsiTwo}}^{2}{
 \gamma}+144\,{{ \PsiTwo}}^{2}{{
 \SigMinLam}}^{2}+144\,{ 
\SigPluLam}\,{{ \PsiTwo}}^{2}{ 
\SigMinLam}+24\,{\Lambda}^{2}{ 
\PsiTwo}+72\,\Lambda\,{{ \PsiTwo}}^{2
}-144\,\Lambda\,{ \gamma}\,{ 
\DeltaPsiTwo}+48\,\Lambda\,{ 
\deltaPsiTwo}\,{ \pi}-48\,\Lambda\,{
 \SigMinLam}\,{ \DeltaPsiTwo}+24\,
\Lambda\,{ \SigPluLam}\,{ 
\DeltaPsiTwo}-864\,{ \PsiTwo}\,{ 
\gamma}\,{ \DeltaPsiTwo}+288\,{ 
\PsiTwo}\,{ \deltaPsiTwo}\,{ 
\pi}-288\,{ \PsiTwo}\,{ 
\SigMinLam}\,{ \DeltaPsiTwo}+144\,{
 \SigPluLam}\,{ \PsiTwo}\,{ 
\DeltaPsiTwo}+24\,\Lambda\,{ 
\DeltaDeltaPsiTwo}+24\,\Lambda\,{ 
\deltadeltaPsiTwo}+144\,{ \PsiTwo}\,{
 \DeltaDeltaPsiTwo}+144\,{ \PsiTwo
}\,{ \deltadeltaPsiTwo}
	\label{}
\end{dmath}
 which has 23 terms.
After some calculation, we find that each of the three candidates from this set are algebraically dependent on $I_{4,a,1}$.
In particular, we have
\begin{align}
	I_{6,b,2}
	&=
	I_{6,b,3}
	=
	\bigl( \Lambda+6\Psi_2 \bigr) I_{4,a,1}
	+ \text{polynomial in lower order invariants}
	\\
	I_{6,b,4}
	&=
	-\frac{1}{4} I_{6,b,2} - \frac{3}{2} \Psi_2 I_{4a1}
	+ \text{polynomial in lower order invariants}
	.
	\label{}
\end{align}
Note that it is permissible to multiply by lower order extended Cartan invariants as these were previously shown equivalent to PSCIs.
Additive terms in lower order extended Cartan invariants have no effect on our goal for the same reason.

We also find that the invariants in $\mathcal{R}^0_{\{2,0,0\}}$ are algebraically dependent on $I_{4,a,1}$.
The explicit dependence is given by 
\begin{align}
	I_{8,e,5}
	&=
	I_{8,e,6}
	=
	I_{8,e,7}
	\\
	&=
	2 \, I_{8,e,8}
	=
	\bigl( \Lambda + 6 \Psi_2)^2 I_{4a1}
	+ \text{polynomial in lower order invariants}
	\nonumber
	.
	\label{}
\end{align}
Note that the specific lower order terms appearing are different for each of the above equalities.

Now the invariants in $\mathcal{R}^0_{\{2,2\}}$ are all that remain of those explicitly presented in FKWC.
The expressions for these invariants are much longer again; see Appendix~\ref{sec:ExampleSecondOrderInvariants} for some explicit expressions.
We find that of the candidates $I_{8,d,2}$ to $I_{8,d,7}$, some are algebraically dependent.
In particular, we have
\begin{align}
	I_{8,d,6}
	&=
	3\Lambda\Psi_2 I_{4a1} + \frac{1}{2} I_{8d3} - \frac{1}{2} I_{8d4} + \frac{1}{2} I_{8d5}
	+ \text{lower order invariants}
	\\
	I_{8,d,7}
	&=
	I_{8d3} - \frac{1}{4} {I_{4a1}}^2
	+ \text{lower order invariants}
	\label{}
\end{align}
Now there are only five remaining candidates for algebraically independent invariants at second order: 1 from $\mathcal{R}^0_{4,1}$ and 4 from $\mathcal{R}^0_{\{2,2\}}$.
The surviving invariants from $\mathcal{R}^0_{\{2,2\}}$ are
\begin{align}
	I_{8,d,2}
	&=
	R^{;pq} R_{;pq}
	\\
	I_{8,d,3}
	&=
	R^{;pq} \Box R_{pq}
	\\
	I_{8,d,4}
	&=
	\Box R^{pq} \Box R_{pq}
	\\
	I_{8,d,5}
	&=
	R_{pq;rs} R^{pq;rs}
	.
	\label{}
\end{align}
However, as mentioned above, FKWC is not a complete classification; there is the possibility of invariants with derivatives higher than 8th order.
For example, we may use
\begin{align}
	I_{10} := 
	R^{;ac} R^{;b}_{\phantom{;b};c} R_{ab}
	\label{}
\end{align}
as another invariant.
It is 10th order, so does not explicitly appear in FKWC.
The name $I_{10}$ is assigned for convenience and consistency.
Another option would be to work with $\Box \Box R$. However, it is not expected to be much easier to work with as it has even more terms than $I_{10}$.

See Table~\ref{tab:SecondOrderInvs} for a list of the remaining candidate invariants and which of the extended Cartan invariants they contain.
It can already be seen from the dependences on different variables that at least three of these invariants are algebraically independent.

\begin{table}
	\centering
	\begin{tabular}{| l | l | c | c | c | c | c | c |}
		\hline
		& terms
		& $\pi$
		& $\sigma+\lambda$
		& $\delta(\sigma-\lambda)$
		& $\delta\delta\Psi_2$
		& $\Delta\Delta\Psi_2$
		& $\Delta\delta\Psi_2$
		\\
		\hline

		$I_{4,a,1}$
		& 11
		& \checkmark 
		& \checkmark
		&
		& \checkmark
		& \checkmark
		&
		\\
		\hline

		$I_{8,d,2}$
		& 65
		& \checkmark
		& \checkmark
		& \checkmark
		& \checkmark
		& \checkmark
		& \checkmark
		\\
		\hline

		$I_{8,d,3}$
		& 74
		& \checkmark
		& \checkmark
		& \checkmark
		& \checkmark
		& \checkmark
		&
		\\
		\hline

		$I_{8,d,4}$
		& 69
		& \checkmark
		& \checkmark
		& \checkmark
		& \checkmark
		& \checkmark
		&
		\\
		\hline

		$I_{8,d,5}$
		& 71
		& \checkmark
		& \checkmark
		& \checkmark
		& \checkmark
		& \checkmark
		& \checkmark
		\\
		\hline

		$I_{10}$
		& 138
		& \checkmark
		& \checkmark
		& \checkmark
		& \checkmark
		& \checkmark
		& \checkmark
		\\
		\hline

	\end{tabular}
	\caption{The polynomial scalar curvature invariants remaining after eliminating those that are clearly algebraically dependent. In the second column is the number of terms in each polynomial. In the remaining columns are indications as to whether the respective extended Cartan invariant appears in the polynomial invariant.}
	\label{tab:SecondOrderInvs}
\end{table}

It is possible to make algebraic combinations of these 6 invariants in such a way that such independence is clearer.
For example, see Table~\ref{tab:SecondOrderInvsSimp}.
After extensive consideration, it is believed that the six invariants presented in Table~\ref{tab:SecondOrderInvs} (or equivalently in Table~\ref{tab:SecondOrderInvsSimp}) are algebraically independent.
\begin{table}
	\centering
	\begin{tabular}{| l | l | c | c | c | c | c | c |}
		\hline
		& $\pi$
		& $\sigma+\lambda$
		& $\delta(\sigma-\lambda)$
		& $\delta\delta\Psi_2$
		& $\Delta\Delta\Psi_2$
		& $\Delta\delta\Psi_2$
		\\
		\hline

		$I_{4,a,1}$
		& \checkmark 
		& \checkmark
		&
		& \checkmark
		& \checkmark
		&
		\\
		\hline

		$I_{8,d,2}$
		& \checkmark
		& \checkmark
		& \checkmark
		& \checkmark
		& \checkmark
		& \checkmark
		\\
		\hline

		$I_{8,d,3}$
		& \checkmark
		& \checkmark
		& \checkmark
		& \checkmark
		& \checkmark
		&
		\\
		\hline

		$I_{8,d,4}$
		& \checkmark
		& \checkmark
		& \checkmark
		& \checkmark
		& \checkmark
		&
		\\
		\hline

%

		$I_{8,d,5} - \frac{1}{2} I_{8,d,2}$
		& \checkmark
		& \checkmark
		& \checkmark
		& \checkmark
		& \checkmark
		&
		\\
		\hline

		$ I_{10} - (\Lambda+6\Psi_2) I_{8d2}$
		& \checkmark
		& \checkmark
		& \checkmark
		&
		&
		&
		\\
		\hline

	\end{tabular}
	\caption{
		The dependence of some algebraic combinations of the polynomial scalar invariants presented in Table~\ref{tab:SecondOrderInvs} on the extended Cartan invariants appearing at second order.
	}
	\label{tab:SecondOrderInvsSimp}
\end{table}

\section{Conclusions}

We are interested in constructing  a minimal
set of algebraically independent scalar curvature invariants formed by the contraction of the Riemann (Ricci) tensor and its
covariant derivatives up to some maximum order of differentiation $q$ in 3D
Lorentzian spacetimes.  In order to do this we have exploited
the CKA since, in general (and in particular for the class of spacetimes considered here), all
Cartan invariants are determined by scalar polynomial curvature invariants  \cite{inv}.

We applied the CKA to the class of  3D {{Szekeres}}
spacetimes, which represent cosmological solutions with comoving
dust and cosmological constant $\Lambda$ \cite{BarrowEtAl}.
In this example, the Segre Type of the traceless Ricci tensor is then $\{ 1(1,1)\}$, 
or equivalently is of `Petrov'-type DZ \cite{MW2013},
$q=2$, and the invariant count is 
(1,3,3) (where $\dim H_0=1$, $\dim H_1=0$, $\dim H_2=0$).

We found that there are at most twelve
algebraically independent extended Cartan invariants including $\Lambda$. 
We computed and presented  our choice of algebraically independent extended Cartan invariants; namely,  
$\Psi_2$ ($\Lambda$) (0th order),
$\gamma$, $\sigma-\lambda$, $\delta\Psi_2$, $\Delta\Psi_2$ (1st order), and
$\pi$, $\sigma+\lambda$, $\delta(\sigma-\lambda)$, $\delta\delta\Psi_2$, $\Delta\Delta\Psi_2$, $\Delta\delta\Psi_2$
(2nd order).

Next we related the above extended Cartan invariants to twelve independent polynomial 
scalar curvature invariants \cite{Coley:2014goa}:
the two zeroth  order  polynomial scalar curvature invariants
$R, R_{ab}R^{ab}$, 
four first order  polynomial scalar curvature invariants (and especially 
$I_{R1}, I_{R2}$) given in subsection (3.2), and
the six second order  polynomial scalar curvature invariants
of subsection (3.3).

The CKA stops at second order, and since we have been
able to find the correspondence between the extended Cartan invariants and the
polynomial invariants presented, these polynomial invariants should be necessary and
sufficient for the classification of the 3D Szekeres spaces.  
That is, these twelve independent polynomial 
scalar curvature invariants, together with information on the
relationships between the invariants
and all of the vanishing invariants, constitute the basis of polynomial 
scalar curvature invariants of the class of 3D {{Szekeres}}
spacetimes. We 
illustrate the form of the invariants for a special case in Appendix B.

{\em Acknowledgements}. We would like to thank Malcolm MacCallum and 
Robert Milson for helpful comments.  
This research was supported by the Natural Sciences and
Engineering Research Council of Canada and by the Perimeter Institute
for Theoretical Physics. Research at
Perimeter Institute is supported by the Government of
Canada through Industry Canada and by the Province
of Ontario through the Ministry of Research and Innovation.

\appendix

\clearpage
\section{Second order invariants}
\label{sec:ExampleSecondOrderInvariants}

In this appendix we demonstrate some of the expressions for the second order polynomial scalar curvature invariants which were too long to include in the main text.

\allowdisplaybreaks

\begin{dmath}
	I_{8,d,2}
	=
	R^{;pq} R_{;pq}
	=
55296\,{{ \PsiTwo}}^{2}{{ \gamma}}
^{2}{{ \pi}}^{2}+6912\,{{ \PsiTwo}
}^{2}{{ \gamma}}^{2}{{ \SigMinLam}
}^{2}-13824\,{{ \PsiTwo}}^{2}{{ 
\gamma}}^{2}{ \SigMinLam}\,{ 
\SigPluLam}+6912\,{{ \PsiTwo}}^{2}{{
 \gamma}}^{2}{{ \SigPluLam}}^{2}+
18432\,{{ \PsiTwo}}^{2}{ \gamma}\,
{ \SigMinLam}\,{{ \pi}}^{2}-2304\,
{{ \PsiTwo}}^{2}{ \gamma}\,{{ 
\SigMinLam}}^{3}+4608\,{{ \PsiTwo}}^{
2}{ \gamma}\,{{ \SigMinLam}}^{2}{
 \SigPluLam}-2304\,{{ \PsiTwo}}^{2
}{ \gamma}\,{ \SigMinLam}\,{{ 
\SigPluLam}}^{2}-4608\,{{ \PsiTwo}}^{
2}{{ \SigMinLam}}^{2}{{ \pi}}^{2}+
1728\,{{ \PsiTwo}}^{2}{{ 
\SigMinLam}}^{4}-3456\,{{ \PsiTwo}}^{
2}{{ \SigMinLam}}^{3}{ \SigPluLam}
+1728\,{{ \PsiTwo}}^{2}{{ 
\SigMinLam}}^{2}{{ \SigPluLam}}^{2}+
2304\,\Lambda\,{ \gamma}\,{{ 
\PsiTwo}}^{2}{ \SigMinLam}-2304\,
\Lambda\,{ \gamma}\,{{ \PsiTwo}}^{
2}{ \SigPluLam}-1152\,\Lambda\,{{ 
\PsiTwo}}^{2}{{ \SigMinLam}}^{2}+1152
\,\Lambda\,{{ \PsiTwo}}^{2}{ 
\SigMinLam}\,{ \SigPluLam}-9216\,{{
 \PsiTwo}}^{2}{ \gamma}\,{ 
\pi}\,{ \deltaSigMinLam}-4608\,{{ 
\PsiTwo}}^{2}{ \SigMinLam}\,{ 
\pi}\,{ \deltaSigMinLam}+18432\,{ 
\PsiTwo}\,{{ \gamma}}^{3}{ 
\DeltaPsiTwo}+9216\,{ \PsiTwo}\,{{
 \gamma}}^{2}{ \DeltaPsiTwo}\,{
 \SigPluLam}-36864\,{ \PsiTwo}\,{
 \gamma}\,{ \DeltaPsiTwo}\,{{ 
\pi}}^{2}-9216\,{ \PsiTwo}\,{ 
\gamma}\,{ \SigMinLam}\,{ 
\pi}\,{ \deltaPsiTwo}-9216\,{ 
\PsiTwo}\,{ \gamma}\,{ 
\SigPluLam}\,{ \pi}\,{ 
\deltaPsiTwo}+4608\,{ \PsiTwo}\,{ 
\gamma}\,{{ \SigMinLam}}^{2}{ 
\DeltaPsiTwo}+2304\,{ \PsiTwo}\,{ 
\gamma}\,{ \SigMinLam}\,{ 
\DeltaPsiTwo}\,{ \SigPluLam}-2304\,{
 \PsiTwo}\,{ \gamma}\,{ 
\DeltaPsiTwo}\,{{ \SigPluLam}}^{2}-
9216\,{ \PsiTwo}\,{{ \SigMinLam}}^
{2}{ \pi}\,{ \deltaPsiTwo}+4608\,{
 \PsiTwo}\,{ \SigMinLam}\,{ 
\SigPluLam}\,{ \pi}\,{ 
\deltaPsiTwo}+2304\,{ \PsiTwo}\,{{
 \SigMinLam}}^{3}{ \DeltaPsiTwo}-
1152\,{ \PsiTwo}\,{{ \SigMinLam}}^
{2}{ \DeltaPsiTwo}\,{ \SigPluLam}-
1152\,{ \PsiTwo}\,{ \SigMinLam}\,{
 \DeltaPsiTwo}\,{{ \SigPluLam}}^{2
}+288\,{\Lambda}^{2}{{ \PsiTwo}}^{2}-3456\,\Lambda\,
{ \gamma}\,{ \PsiTwo}\,{ 
\DeltaPsiTwo}+1152\,\Lambda\,{ 
\PsiTwo}\,{ \deltaPsiTwo}\,{ 
\pi}-1152\,\Lambda\,{ \PsiTwo}\,{ 
\SigMinLam}\,{ \DeltaPsiTwo}+9216\,{
 \PsiTwo}\,{{ \gamma}}^{2}{ 
\DeltaDeltaPsiTwo}+18432\,{ \PsiTwo}
\,{ \gamma}\,{ \pi}\,{ 
\DeltadeltaPsiTwo}+9216\,{ \PsiTwo}\,
{ \gamma}\,{ \SigMinLam}\,{ 
\DeltaDeltaPsiTwo}+2304\,{ \PsiTwo}\,
{ \gamma}\,{ \SigMinLam}\,{ 
\deltadeltaPsiTwo}-4608\,{ \PsiTwo}\,
{ \gamma}\,{ \SigPluLam}\,{ 
\DeltaDeltaPsiTwo}-2304\,{ \PsiTwo}\,
{ \gamma}\,{ \SigPluLam}\,{ 
\deltadeltaPsiTwo}+4608\,{ \PsiTwo}\,
{ \DeltaPsiTwo}\,{ \pi}\,{ 
\deltaSigMinLam}+1152\,{ \PsiTwo}\,{{
 \SigMinLam}}^{2}{ 
\deltadeltaPsiTwo}+2304\,{ \PsiTwo}\,
{ \SigMinLam}\,{ \SigPluLam}\,{
 \DeltaDeltaPsiTwo}-1152\,{ 
\PsiTwo}\,{ \SigMinLam}\,{ 
\SigPluLam}\,{ \deltadeltaPsiTwo}+
2304\,{{ \gamma}}^{2}{{ 
\DeltaPsiTwo}}^{2}-4608\,{ \gamma}\,{
 \DeltaPsiTwo}\,{ \pi}\,{ 
\deltaPsiTwo}+4608\,{ \gamma}\,{ 
\SigMinLam}\,{{ \DeltaPsiTwo}}^{2}+
4608\,{{ \DeltaPsiTwo}}^{2}{{ \pi}
}^{2}+2304\,{{ \pi}}^{2}{{ 
\deltaPsiTwo}}^{2}+4608\,{ \SigMinLam
}\,{ \DeltaPsiTwo}\,{ \pi}\,{ 
\deltaPsiTwo}+1152\,{{ \SigMinLam}}^{
2}{{ \DeltaPsiTwo}}^{2}+576\,{{ 
\DeltaPsiTwo}}^{2}{{ \SigPluLam}}^{2}
+576\,\Lambda\,{ \PsiTwo}\,{ 
\DeltaDeltaPsiTwo}-2304\,{ \PsiTwo}\,
{ \DeltadeltaPsiTwo}\,{ 
\deltaSigMinLam}-6912\,{ \gamma}\,{
 \DeltaPsiTwo}\,{ 
\DeltaDeltaPsiTwo}-4608\,{ \gamma}\,{
 \deltaPsiTwo}\,{ 
\DeltadeltaPsiTwo}-4608\,{ 
\DeltaPsiTwo}\,{ \pi}\,{ 
\DeltadeltaPsiTwo}+2304\,{ \pi}\,{
 \deltaPsiTwo}\,{ 
\DeltaDeltaPsiTwo}-2304\,{ \SigMinLam
}\,{ \DeltaPsiTwo}\,{ 
\DeltaDeltaPsiTwo}+1152\,{ 
\DeltaPsiTwo}\,{ \SigPluLam}\,{ 
\deltadeltaPsiTwo}-4608\,{ \SigMinLam
}\,{ \deltaPsiTwo}\,{ 
\DeltadeltaPsiTwo}+576\,{{ 
\DeltaDeltaPsiTwo}}^{2}+1152\,{{ 
\DeltadeltaPsiTwo}}^{2}+576\,{{ 
\deltadeltaPsiTwo}}^{2}
\end{dmath}

\begin{dmath}
	I_{8,d,3}
	=
	R^{;pq} \Box R_{pq}
	=
-4608\,{{ \gamma}}^{3}{{ \PsiTwo}}
^{2}{ \SigMinLam}-4608\,{{ \gamma}
}^{3}{{ \PsiTwo}}^{2}{ \SigPluLam}
-9216\,{{ \PsiTwo}}^{2}{{ \gamma}}
^{2}{{ \pi}}^{2}+5760\,{{ \PsiTwo}
}^{2}{{ \gamma}}^{2}{{ \SigMinLam}
}^{2}-18432\,{{ \PsiTwo}}^{2}{{ 
\gamma}}^{2}{ \SigMinLam}\,{ 
\SigPluLam}+8064\,{{ \PsiTwo}}^{2}{{
 \gamma}}^{2}{{ \SigPluLam}}^{2}+
9216\,{{ \PsiTwo}}^{2}{ \gamma}\,{
 \SigMinLam}\,{{ \pi}}^{2}-1152\,{
{ \PsiTwo}}^{2}{ \gamma}\,{{ 
\SigMinLam}}^{2}{ \SigPluLam}-1152\,{
{ \PsiTwo}}^{2}{ \gamma}\,{ 
\SigMinLam}\,{{ \SigPluLam}}^{2}-2304
\,{{ \PsiTwo}}^{2}{{ \SigMinLam}}^
{2}{{ \pi}}^{2}+1440\,{{ \PsiTwo}}
^{2}{{ \SigMinLam}}^{4}-1152\,{{ 
\PsiTwo}}^{2}{{ \SigMinLam}}^{3}{ 
\SigPluLam}-288\,{{ \PsiTwo}}^{2}{{
 \SigMinLam}}^{2}{{ \SigPluLam}}^{
2}+2304\,\Lambda\,{{ \gamma}}^{2}{{ 
\PsiTwo}}^{2}+3744\,\Lambda\,{ \gamma
}\,{{ \PsiTwo}}^{2}{ \SigMinLam}-
2592\,\Lambda\,{ \gamma}\,{{ 
\PsiTwo}}^{2}{ \SigPluLam}-144\,
\Lambda\,{{ \PsiTwo}}^{2}{{ 
\SigMinLam}}^{2}+432\,\Lambda\,{{ 
\PsiTwo}}^{2}{ \SigMinLam}\,{ 
\SigPluLam}+4608\,{{ \PsiTwo}}^{2}{
 \gamma}\,{ \pi}\,{ 
\deltaSigMinLam}-2304\,{{ \PsiTwo}}^{
2}{ \SigMinLam}\,{ \pi}\,{ 
\deltaSigMinLam}+9216\,{ \PsiTwo}\,{{
 \gamma}}^{3}{ \DeltaPsiTwo}+13824
\,{{ \gamma}}^{2}{ \pi}\,{ 
\deltaPsiTwo}\,{ \PsiTwo}-9216\,{ 
\DeltaPsiTwo}\,{{ \gamma}}^{2}{ 
\PsiTwo}\,{ \SigMinLam}+13824\,{ 
\PsiTwo}\,{{ \gamma}}^{2}{ 
\DeltaPsiTwo}\,{ \SigPluLam}+16128\,{
 \PsiTwo}\,{ \gamma}\,{ 
\SigMinLam}\,{ \pi}\,{ 
\deltaPsiTwo}-6912\,{ \PsiTwo}\,{ 
\gamma}\,{ \SigPluLam}\,{ 
\pi}\,{ \deltaPsiTwo}-4608\,{ 
\PsiTwo}\,{ \gamma}\,{{ 
\SigMinLam}}^{2}{ \DeltaPsiTwo}+11520
\,{ \PsiTwo}\,{ \gamma}\,{ 
\SigMinLam}\,{ \DeltaPsiTwo}\,{ 
\SigPluLam}-3456\,{ \PsiTwo}\,{ 
\gamma}\,{ \DeltaPsiTwo}\,{{ 
\SigPluLam}}^{2}-4608\,{ \PsiTwo}\,{{
 \SigMinLam}}^{2}{ \pi}\,{ 
\deltaPsiTwo}+1152\,{ \PsiTwo}\,{ 
\SigMinLam}\,{ \SigPluLam}\,{ 
\pi}\,{ \deltaPsiTwo}+576\,{ 
\PsiTwo}\,{ \SigMinLam}\,{ 
\DeltaPsiTwo}\,{{ \SigPluLam}}^{2}+
144\,{\Lambda}^{2}{{ \PsiTwo}}^{2}-2592\,\Lambda\,{
 \gamma}\,{ \PsiTwo}\,{ 
\DeltaPsiTwo}+864\,\Lambda\,{ \PsiTwo
}\,{ \deltaPsiTwo}\,{ \pi}-864\,
\Lambda\,{ \PsiTwo}\,{ \SigMinLam}
\,{ \DeltaPsiTwo}+432\,\Lambda\,{ 
\PsiTwo}\,{ \DeltaPsiTwo}\,{ 
\SigPluLam}-576\,{{ \deltaSigMinLam}}
^{2}{{ \PsiTwo}}^{2}+4608\,{ 
\PsiTwo}\,{{ \gamma}}^{2}{ 
\DeltaDeltaPsiTwo}+2304\,{ 
\deltadeltaPsiTwo}\,{{ \gamma}}^{2}{
 \PsiTwo}+4608\,{ \PsiTwo}\,{ 
\gamma}\,{ \SigMinLam}\,{ 
\DeltaDeltaPsiTwo}+4608\,{ \PsiTwo}\,
{ \gamma}\,{ \SigMinLam}\,{ 
\deltadeltaPsiTwo}-3456\,{ \PsiTwo}\,
{ \gamma}\,{ \SigPluLam}\,{ 
\DeltaDeltaPsiTwo}-3456\,{ \PsiTwo}\,
{ \gamma}\,{ \SigPluLam}\,{ 
\deltadeltaPsiTwo}-1152\,{ 
\deltaSigMinLam}\,{ \gamma}\,{ 
\deltaPsiTwo}\,{ \PsiTwo}+576\,{ 
\PsiTwo}\,{{ \SigMinLam}}^{2}{ 
\deltadeltaPsiTwo}+576\,{ \PsiTwo}\,{
 \SigMinLam}\,{ \SigPluLam}\,{ 
\DeltaDeltaPsiTwo}+576\,{ \PsiTwo}\,{
 \SigMinLam}\,{ \SigPluLam}\,{ 
\deltadeltaPsiTwo}-2304\,{ 
\deltaSigMinLam}\,{ \deltaPsiTwo}\,{
 \PsiTwo}\,{ \SigMinLam}+1152\,{{
 \gamma}}^{2}{{ \DeltaPsiTwo}}^{2}
-6912\,{ \gamma}\,{ \DeltaPsiTwo}
\,{ \pi}\,{ \deltaPsiTwo}+4608\,{
 \gamma}\,{ \SigMinLam}\,{{ 
\DeltaPsiTwo}}^{2}-3456\,{{ 
\DeltaPsiTwo}}^{2}{ \gamma}\,{ 
\SigPluLam}-2304\,{ \gamma}\,{{ 
\deltaPsiTwo}}^{2}{ \SigMinLam}+1152
\,{{ \pi}}^{2}{{ \deltaPsiTwo}}^{2
}-2304\,{ \SigMinLam}\,{ 
\DeltaPsiTwo}\,{ \pi}\,{ 
\deltaPsiTwo}+1152\,{ \DeltaPsiTwo}\,
{ \pi}\,{ \deltaPsiTwo}\,{ 
\SigPluLam}+1152\,{{ \SigMinLam}}^{2}
{{ \DeltaPsiTwo}}^{2}-1152\,{{ 
\DeltaPsiTwo}}^{2}{ \SigMinLam}\,{
 \SigPluLam}+288\,{{ \DeltaPsiTwo}
}^{2}{{ \SigPluLam}}^{2}-2304\,{{ 
\deltaPsiTwo}}^{2}{{ \SigMinLam}}^{2}
+432\,\Lambda\,{ \PsiTwo}\,{ 
\DeltaDeltaPsiTwo}+432\,\Lambda\,{ 
\PsiTwo}\,{ \deltadeltaPsiTwo}-3456\,
{ \gamma}\,{ \DeltaPsiTwo}\,{ 
\DeltaDeltaPsiTwo}-3456\,{ 
\deltadeltaPsiTwo}\,{ \DeltaPsiTwo}\,
{ \gamma}+1152\,{ \pi}\,{ 
\deltaPsiTwo}\,{ \DeltaDeltaPsiTwo}+
1152\,{ \deltadeltaPsiTwo}\,{ \pi}
\,{ \deltaPsiTwo}-1152\,{ 
\SigMinLam}\,{ \DeltaPsiTwo}\,{ 
\DeltaDeltaPsiTwo}-1152\,{ 
\deltadeltaPsiTwo}\,{ \DeltaPsiTwo}\,
{ \SigMinLam}+576\,{ \DeltaPsiTwo}
\,{ \DeltaDeltaPsiTwo}\,{ 
\SigPluLam}+576\,{ \DeltaPsiTwo}\,{
 \SigPluLam}\,{ \deltadeltaPsiTwo}
+288\,{{ \DeltaDeltaPsiTwo}}^{2}+576\,{ 
\deltadeltaPsiTwo}\,{ 
\DeltaDeltaPsiTwo}+288\,{{ 
\deltadeltaPsiTwo}}^{2}
\end{dmath}

\section{A Special Case}
\label{sec:SpecialCase}

It is illustrative to demonstrate the above relationships by examining a special case of the 3D Szekeres spacetimes. We first motivate our choice of form for the arbitrary function $S(y)$ and those appearing in $\hat{\nu}(x,y) = -\ln(A(y) x^2 + B(y) x + C(y))$.

\subsection{$S(y)$}

We desire a special case in which the Cartan invariants take on a simpler (but nondegenerate) form. They can then be more easily compared with the polynomial scalar curvature invariants explicitly.

$F_{,x}=0$ for our Szekeres metrics, so (as in \eqref{BarrowE}) equations (87)-(88) of Barrow et al. become
\begin{equation}
	\hat{\kappa} \hat{\rho}
	= \frac{E(x,y; \Lambda)}{R(R_{,y} + R \nu_{,y})}
\end{equation}
where
\begin{equation}
	E(x,y; \Lambda) = e^{-2\nu}[\nu_{,xy} \nu_{,x} - \nu_{,xxy} ]
	+ \frac{1}{2} e^{-2\nu}(K e^{2\nu})_{,y}
\end{equation}
and
\begin{equation}
	K(x,y,t) = \dot{R}^2 - \Lambda R^2 - 2 (SF_{,y})^2.
\end{equation}
Defining $S(y)$ to be such that $K$ vanishes, that is,
\begin{equation}
	\label{Sassumption}
	S^2 = \frac{\dot{R}^2 - \Lambda R^2}{ 2 {F_{,y}}^2}  = \frac{-\Lambda + F^2}{ 2 {F_{,y}}^2},
\end{equation}
we have 
\begin{equation}
	\Psi_2
	= \frac{1}{12} \hat{\kappa} \hat{\rho}
	= \frac{1}{12} \frac{e^{-2\nu}[\nu_{,xy} \nu_{,x} - \nu_{,xxy} ]}{R(R_{,y} + R \nu_{,y})},
\end{equation}
which is simpler. Note that due to a differential identity of $R(t,y)$, \eqref{Sassumption} is in fact not a function of $t$. Also note that we must assume $-\Lambda + F^2 \gneq 0$.

\subsection{$\hat{\nu}(x,y)$}

We have 
\begin{equation}
 \hat{\nu} = -\ln(A(y) x^2 + B(y) x + C(y)).
\end{equation}
Choosing only $B(y)=0$ and $C(y)=c$ constant simplifies the derivatives of $\hnu$ without causing the 
$\Psi_2$ above to vanish.
This implies that only two arbitrary functions remain: $A(y)$ and $F(y)$.

\subsection{Extended Cartan Invariants}

After making these choices, the Cartan invariants take the forms $\Lambda$ constant and:
\begin{dmath}
	\Psi_2 
	=
1/6\,{\frac {A_{{y}}c \left( A{x}^{2}-c \right) }{R \left( -R_{{y}}A{x
}^{2}+RA_{{y}}{x}^{2}-R_{{y}}c \right) }}
\end{dmath}
\begin{dmath}
	\gamma 
	=
-1/4\,{\frac {R_{{t}}}{R}}
\end{dmath}
\begin{dmath}
	\frac{\sigma-\lambda}{2}
	=
-1/4\,{\frac {-R_{{t,y}}A{x}^{2}+R_{{t}}A_{{y}}{x}^{2}-R_{{t,y}}c}{-R_
{{y}}A{x}^{2}+RA_{{y}}{x}^{2}-R_{{y}}c}}
\end{dmath}
\begin{dmath}
	\delta\Psi_2
	=
	\frac{-\sqrt {-\Lambda\,{R}^{2}+{R_{{t}}}^{2}}c \left( A{x}^{2}+c \right)}{6\, \left( -R_{{y}}A{x}^{2}+RA_{{y}}{x}^{2}-R_{{y}}c \right) ^{3}{R}^{
2}
}
 \Bigl( A_{{y}}{A}^{2}R_{{y,y
}}R{x}^{4}-{A}^{2}R_{{y}}A_{{y,y}}R{x}^{4}+A_{{y}}{A}^{2}{R_{{y}}}^{2}
{x}^{4}-2\,{A_{{y}}}^{2}AR_{{y}}R{x}^{4}+{A_{{y}}}^{3}{R}^{2}{x}^{4}-A
_{{y}}R_{{y,y}}R{c}^{2}+R_{{y}}A_{{y,y}}R{c}^{2}-A_{{y}}{R_{{y}}}^{2}{
c}^{2} \Bigr)  
\end{dmath}
\begin{dmath}
	\Delta\Psi_2
	=
	\frac{-A_{{y}}c }{12\,{R}^{2} \left( -R_{{y}}A{x}^{2}+RA_{{y}}{x}^{2}-R_{{y}}c \right) ^{2}}
\Bigl( -R_{{t,y}}R{A}^{2}{x}^{4}-R_{{t}}R_{{y}}{A}^{2}{x}^{4
}+2\,A_{{y}}R_{{t}}RA{x}^{4}+4\,R_{{y}}{A}^{2}c{x}^{3}-2\,A_{{y}}RAc{x
}^{3}-2\,A_{{y}}R_{{t}}Rc{x}^{2}+4\,R_{{y}}A{c}^{2}x-2\,A_{{y}}R{c}^{2
}x+R_{{t,y}}R{c}^{2}+R_{{t}}R_{{y}}{c}^{2} \Bigr) 
\end{dmath}
\begin{dmath}
	\pi
	=
1/4\,{\frac {\sqrt {-\Lambda\,{R}^{2}+{R_{{t}}}^{2}}}{R}}
\end{dmath}
\begin{dmath}
	\frac{\sigma+\lambda}{2}
	=
1/2\,{\frac {A_{{y}}xc}{-R_{{y}}A{x}^{2}+RA_{{y}}{x}^{2}-R_{{y}}c}}
\end{dmath}
\begin{dmath}
	\delta(\sigma-\lambda)
	=
	\frac{-\sqrt {-\Lambda\,{R}^{2}+{R_{{t}}}^{2}}\left( A{x}^{2}+c \right) }{4\, \left( -R_{{y}}A{x}^{2}+RA_{{y}}{x}^{2}-R_{{y}}c \right) ^{3}}
 \Bigl( -{A}^{2}R_{{y}}R_{{y,t
,y}}{x}^{4}+R_{{t,y}}{A}^{2}R_{{y,y}}{x}^{4}+RA_{{y}}AR_{{y,t,y}}{x}^{
4}-R_{{t,y}}RAA_{{y,y}}{x}^{4}-R_{{t}}A_{{y}}AR_{{y,y}}{x}^{4}+R_{{t}}
AR_{{y}}A_{{y,y}}{x}^{4}-2\,AR_{{y}}R_{{y,t,y}}c{x}^{2}+2\,R_{{t,y}}AR
_{{y,y}}c{x}^{2}+RA_{{y}}R_{{y,t,y}}c{x}^{2}-R_{{t,y}}RA_{{y,y}}c{x}^{
2}-R_{{t}}A_{{y}}R_{{y,y}}c{x}^{2}+R_{{t}}R_{{y}}A_{{y,y}}c{x}^{2}-R_{
{y}}R_{{y,t,y}}{c}^{2}+R_{{t,y}}R_{{y,y}}{c}^{2} \Bigr)  
\end{dmath}
\begin{dmath}
	\delta\delta\Psi_2
	=
	\text{196 terms}
\end{dmath}
\begin{dmath}
	\Delta\Delta\Psi_2
	=
	\frac{1}{24\,{R}^{3} \left( -R_{{y}}A{x}^{2}+RA_{{y}}{x}^{2}-R_{{y}}c \right) ^
{3}}
\Bigl(
-A_{{y}}c \left( 4\,{A}^{4}{R_{{y}}}^{2}c{x}^{6}-2\,{R}^{2}{A}^{3}{R_{
{t,y}}}^{2}{x}^{6}-6\,A_{{y}}R{A}^{3}R_{{y}}c{x}^{6}-2\,R{A}^{3}R_{{t}
}R_{{y}}R_{{t,y}}{x}^{6}+R_{{t,t}}R{A}^{3}{R_{{y}}}^{2}{x}^{6}-2\,{A}^
{3}{R_{{t}}}^{2}{R_{{y}}}^{2}{x}^{6}+2\,{A_{{y}}}^{2}{R}^{2}{A}^{2}c{x
}^{6}+6\,A_{{y}}{R}^{2}{A}^{2}R_{{t}}R_{{t,y}}{x}^{6}-3\,A_{{y}}R_{{t,
t}}{R}^{2}{A}^{2}R_{{y}}{x}^{6}+6\,A_{{y}}R{A}^{2}{R_{{t}}}^{2}R_{{y}}
{x}^{6}+2\,{A_{{y}}}^{2}R_{{t,t}}{R}^{3}A{x}^{6}-6\,{A_{{y}}}^{2}{R}^{
2}A{R_{{t}}}^{2}{x}^{6}+8\,R{A}^{3}R_{{y}}R_{{t,y}}c{x}^{5}+12\,{A}^{3
}R_{{t}}{R_{{y}}}^{2}c{x}^{5}-30\,A_{{y}}R{A}^{2}R_{{t}}R_{{y}}c{x}^{5
}+10\,{A_{{y}}}^{2}{R}^{2}AR_{{t}}c{x}^{5}+4\,{A}^{3}{R_{{y}}}^{2}{c}^
{2}{x}^{4}-2\,{R}^{2}{A}^{2}{R_{{t,y}}}^{2}c{x}^{4}-18\,A_{{y}}R{A}^{2
}R_{{y}}{c}^{2}{x}^{4}-2\,R{A}^{2}R_{{t}}R_{{y}}R_{{t,y}}c{x}^{4}+R_{{
t,t}}R{A}^{2}{R_{{y}}}^{2}c{x}^{4}-2\,{A}^{2}{R_{{t}}}^{2}{R_{{y}}}^{2
}c{x}^{4}+8\,{A_{{y}}}^{2}{R}^{2}A{c}^{2}{x}^{4}-2\,{A_{{y}}}^{2}R_{{t
,t}}{R}^{3}c{x}^{4}+6\,{A_{{y}}}^{2}{R}^{2}{R_{{t}}}^{2}c{x}^{4}+16\,R
{A}^{2}R_{{y}}R_{{t,y}}{c}^{2}{x}^{3}+24\,{A}^{2}R_{{t}}{R_{{y}}}^{2}{
c}^{2}{x}^{3}-8\,A_{{y}}{R}^{2}AR_{{t,y}}{c}^{2}{x}^{3}-32\,A_{{y}}RAR
_{{t}}R_{{y}}{c}^{2}{x}^{3}+10\,{A_{{y}}}^{2}{R}^{2}R_{{t}}{c}^{2}{x}^
{3}-4\,{A}^{2}{R_{{y}}}^{2}{c}^{3}{x}^{2}+2\,{R}^{2}A{R_{{t,y}}}^{2}{c
}^{2}{x}^{2}-10\,A_{{y}}RAR_{{y}}{c}^{3}{x}^{2}+2\,RAR_{{t}}R_{{y}}R_{
{t,y}}{c}^{2}{x}^{2}-R_{{t,t}}RA{R_{{y}}}^{2}{c}^{2}{x}^{2}+2\,A{R_{{t
}}}^{2}{R_{{y}}}^{2}{c}^{2}{x}^{2}+6\,{A_{{y}}}^{2}{R}^{2}{c}^{3}{x}^{
2}-6\,A_{{y}}{R}^{2}R_{{t}}R_{{t,y}}{c}^{2}{x}^{2}+3\,A_{{y}}R_{{t,t}}
{R}^{2}R_{{y}}{c}^{2}{x}^{2}-6\,A_{{y}}R{R_{{t}}}^{2}R_{{y}}{c}^{2}{x}
^{2}+8\,RAR_{{y}}R_{{t,y}}{c}^{3}x+12\,AR_{{t}}{R_{{y}}}^{2}{c}^{3}x-8
\,A_{{y}}{R}^{2}R_{{t,y}}{c}^{3}x-2\,A_{{y}}RR_{{t}}R_{{y}}{c}^{3}x-4
\,A{R_{{y}}}^{2}{c}^{4}+2\,{R}^{2}{R_{{t,y}}}^{2}{c}^{3}+2\,A_{{y}}RR_
{{y}}{c}^{4}+2\,RR_{{t}}R_{{y}}R_{{t,y}}{c}^{3}-R_{{t,t}}R{R_{{y}}}^{2
}{c}^{3}+2\,{R_{{t}}}^{2}{R_{{y}}}^{2}{c}^{3} \right) 
\Bigr)
\end{dmath}
\begin{dmath}
	\Delta\delta\Psi_2
	=
	\text{214 terms}
\end{dmath}

\subsection{Scalar polynomial invariants}

Likewise, the explicit polynomial scalar invariants are simplified. 
Let us simply present the two 0th order invariants from section~\ref{sec:ZerothOrderPolynomialInvariants}:
\begin{dmath}
	R
	=
	\frac{1}{R \left( -R_{{y}}A{x}^{2}+RA_{{y}}{x}^{2}-R_{{y}}c \right) }
\Bigl(
-\Lambda\,RAR_{{y}}{x}^{2}+\Lambda\,{R}^{2}A_{{y}}{x}^{2}+2\,A{x}^{2}A
_{{y}}c-R_{{t,t}}AR_{{y}}{x}^{2}+2\,A_{{y}}{x}^{2}RR_{{t,t}}-\Lambda\,
RR_{{y}}c-2\,A_{{y}}{c}^{2}-R_{{t,t}}R_{{y}}c
\Bigr)
\end{dmath}
\begin{dmath}
	R^{ab} R_{ab}
	=
	\frac{1}{2\,{R}^{2} \left( -R_{{y}}A{x}^{2}+A_{{y}}{x}^{2}R-R_{{y}}c \right) ^{2}}
\Bigl(
{\Lambda}^{2}{R}^{2}{A}^{2}{R_{{y}}}^{2}{x}^{4}-2\,{\Lambda}^{2}{R}^{3
}A_{{y}}AR_{{y}}{x}^{4}+{\Lambda}^{2}{R}^{4}{A_{{y}}}^{2}{x}^{4}-4\,
\Lambda\,RA_{{y}}{A}^{2}R_{{y}}c{x}^{4}+\Lambda\,R{A}^{2}{R_{{y}}}^{2}
R_{{t,t}}{x}^{4}+4\,\Lambda\,{R}^{2}{A_{{y}}}^{2}Ac{x}^{4}-3\,\Lambda
\,{R}^{2}A_{{y}}AR_{{y}}R_{{t,t}}{x}^{4}+2\,\Lambda\,{R}^{3}{A_{{y}}}^
{2}R_{{t,t}}{x}^{4}+4\,{A_{{y}}}^{2}{A}^{2}{c}^{2}{x}^{4}-2\,A_{{y}}{A
}^{2}R_{{y}}R_{{t,t}}c{x}^{4}+{A}^{2}{R_{{y}}}^{2}{R_{{t,t}}}^{2}{x}^{
4}+2\,{\Lambda}^{2}{R}^{2}A{R_{{y}}}^{2}c{x}^{2}+4\,R{A_{{y}}}^{2}AR_{
{t,t}}c{x}^{4}-3\,RA_{{y}}AR_{{y}}{R_{{t,t}}}^{2}{x}^{4}-2\,{\Lambda}^
{2}{R}^{3}A_{{y}}R_{{y}}c{x}^{2}+3\,{R}^{2}{A_{{y}}}^{2}{R_{{t,t}}}^{2
}{x}^{4}+2\,\Lambda\,RA{R_{{y}}}^{2}R_{{t,t}}c{x}^{2}-4\,\Lambda\,{R}^
{2}{A_{{y}}}^{2}{c}^{2}{x}^{2}-3\,\Lambda\,{R}^{2}A_{{y}}R_{{y}}R_{{t,
t}}c{x}^{2}-8\,{A_{{y}}}^{2}A{c}^{3}{x}^{2}+2\,A{R_{{y}}}^{2}{R_{{t,t}
}}^{2}c{x}^{2}+{\Lambda}^{2}{R}^{2}{R_{{y}}}^{2}{c}^{2}-4\,R{A_{{y}}}^
{2}R_{{t,t}}{c}^{2}{x}^{2}-3\,RA_{{y}}R_{{y}}{R_{{t,t}}}^{2}c{x}^{2}+4
\,\Lambda\,RA_{{y}}R_{{y}}{c}^{3}+\Lambda\,R{R_{{y}}}^{2}R_{{t,t}}{c}^
{2}+4\,{A_{{y}}}^{2}{c}^{4}+2\,A_{{y}}R_{{y}}R_{{t,t}}{c}^{3}+{R_{{y}}
}^{2}{R_{{t,t}}}^{2}{c}^{2}
\Bigr)
\end{dmath}
Using these expressions and those above, we verify that \eqref{eqn:ZerothOrderPSCIs} holds in the special case.
The higher order polynomial invariants can also be calculated in the special case. 
However, as one may imagine from the expressions for the 1st and 2nd order Cartan invariants above, the expressions are somewhat unwieldy.

\end{document}